\documentclass[a4paper,USenglish,cleveref, autoref, thm-restate]{lipics-v2021}

\usepackage{environ}
\newcounter{asyfigcntr}
\makeatletter
\@ifundefined{asy}{%
  \NewEnviron{asy}{%
    \par\stepcounter{asyfigcntr}%
    \includegraphics{output-\theasyfigcntr.pdf}%
    \par%
  }
}{}
\makeatother

\usepackage{todonotes}
\usepackage{thmtools}
\usepackage{tablefootnote}
\allowdisplaybreaks
\usepackage[subtle, mathspacing=normal]{savetrees}

\DeclareMathOperator{\closest}{\mathit{closest}}
\DeclareMathOperator{\unit}{unit}
\DeclareMathOperator{\spacing}{\mathit{spacing}}

\DeclareMathOperator{\vol}{vol}

\newcommand{\R}{\mathbb{R}}
\newcommand{\eps}{\epsilon}
\newcommand{\MO}{\mathcal O}

\newcommand{\vmag}[1]{\left\| #1 \right\|}
\newcommand{\paren}[1]{\left( #1 \right)}
\newcommand{\BIGO}[1]{\MO\left(#1\right)}

\newcommand{\ceil}[1]{\left\lceil #1 \right\rceil}

\crefformat{footnote}{#2\footnotemark[#1]#3}

\hideLIPIcs  

\title{New Approximation Algorithms for Touring Regions} 

\author{Benjamin Qi\footnote{\label{note1}These authors contributed equally.}}{Massachusetts Institute of Technology, Cambridge, MA, USA}{bqi343@gmail.com}{https://orcid.org/0000-0002-0721-2036}{}

\author{Richard Qi\cref{note1}}{Massachusetts Institute of Technology, Cambridge, MA, USA}{rqi@mit.edu}{}{}

\author{Xinyang Chen}{Massachusetts Institute of Technology, Cambridge, MA, USA}{chenxy@mit.edu}{}{}

\authorrunning{B. Qi, R. Qi, and X. Chen}

\Copyright{Benjamin Qi, Richard Qi, and Xinyang Chen} 



\category{} 

\relatedversion{This is the full version of a paper that will appear in SoCG 2023.}

\acknowledgements{We thank Dhruv Rohatgi, Quanquan Liu, Erik Demaine, Danny Mittal, Spencer Compton, William Kuszmaul,  Timothy Qian, Jonathan Kelner, and Chris Zhang for helpful discussions.}

\nolinenumbers 

\EventEditors{John Q. Open and Joan R. Access}
\EventNoEds{2}
\EventLongTitle{42nd Conference on Very Important Topics (CVIT 2016)}
\EventShortTitle{CVIT 2016}
\EventAcronym{CVIT}
\EventYear{2016}
\EventDate{December 24--27, 2016}
\EventLocation{Little Whinging, United Kingdom}
\EventLogo{}
\SeriesVolume{42}
\ArticleNo{23}

\begin{document}

\maketitle

\begin{abstract}
We analyze the \emph{touring regions problem}: find a ($1+\epsilon$)-approximate Euclidean shortest path in $d$-dimensional space that starts at a given starting point, ends at a given ending point, and visits given regions $R_1, R_2, R_3, \dots, R_n$ in that order. 

Our main result is an $\mathcal O \left(\frac{n}{\sqrt{\epsilon}}\log{\frac{1}{\epsilon}} + \frac{1}{\epsilon} \right)$-time algorithm for touring disjoint disks. We also give an $\mathcal O\left (\min\left(\frac{n}{\epsilon}, \frac{n^2}{\sqrt \epsilon}\right) \right)$-time algorithm for touring disjoint two-dimensional convex fat bodies. Both of these results naturally generalize to larger dimensions; we obtain $\mathcal O\left(\frac{n}{\epsilon^{d-1}}\log^2\frac{1}{\epsilon}+\frac{1}{\epsilon^{2d-2}}\right)$ and $\mathcal O\left(\frac{n}{\epsilon^{2d-2}}\right)$-time algorithms for touring disjoint $d$-dimensional balls and convex fat bodies, respectively.
\end{abstract}

\section{Introduction}

We analyze the \emph{touring regions problem}: find a ($1+\eps$)-approximate Euclidean shortest path in $d$-dimensional space that starts at a given starting point, ends at a given ending point, and visits given regions $R_1, R_2, R_3, \dots, R_n$ in that order. We present algorithms for the cases where the regions $R_i$ are constrained to be unions of general convex bodies, convex fat bodies, or balls. To the best of our knowledge, we are the first to consider the cases where regions are disjoint convex fat bodies or balls in arbitrary dimensions. Consequently, our algorithms use techniques not previously considered in the touring regions literature (\cref{sec:summary-techniques}). Our algorithms work under the assumption that a closest point oracle is provided; closest point projection has been extensively used and studied in convex optimization and mathematics \cite{doi:10.1137/S0036144593251710,nearest-points}.

Most prior work focuses on $d=2$ or significantly restricts the convex bodies. The special case where $d=2$ and all regions are constrained to be polygons is known as the \emph{touring polygons problem}. Dror et al. \cite{dror2003} solved the case where every region is a convex polygon exactly, presenting an $\BIGO{|V|n \log{\frac{|V|}{n}}}$-time algorithm when the regions are disjoint as well as an $\BIGO{|V|n^2 \log{|V|}}$-time algorithm when the regions are possibly non-disjoint and the subpath between every two consecutive polygons in the tour is constrained to lie within a simply connected region called a \emph{fence}. Here, $|V|$ is the total number of vertices over all polygons. Tan and Jiang \cite{starburst-improve} improved these bounds to $\BIGO{|V|n}$ and $\BIGO{|V|n^2}$-time, respectively, without considering subpath constraints. 

For touring nonconvex polygons, Ahadi et al. \cite{ahadi2014touring} proved that finding an optimal path is NP-hard even when polygons are disjoint and constrained to be two line segments each. Dror et al. \cite{dror2003} showed that approximately touring nonconvex polygons with constraining fences is a special case of 3D shortest path with obstacle polyhedra, which can be solved in $\tilde{\MO}\paren{\frac{e^4}{\epsilon^2}}$ time by applying results of Asano et al. \cite{asano2004pseudo}, where $e$ is the total number of edges over all polyhedra. Mozafari and Zarei \cite{Mozafari2012TouringPA} improved the bound for the case of nonconvex polygons with constraining fences to $\tilde{\MO}\paren{\frac{|V|^2n^2}{\eps^2}}$ time. Ahadi et al. \cite{ahadi2014touring} also solve the \emph{touring objects problem} exactly in polynomial time, in which the $R_i$ are disjoint, nonconvex polygons and the objective is to visit the border of every region without entering the interior of any region.

For touring disjoint disks, a heuristic algorithm with experimental results was demonstrated by Chou \cite{exact-heuristic}. Touring disjoint \emph{unit} disks was given in a programming contest and was a source of inspiration for this paper; an $\BIGO{\frac{n}{\epsilon^2}}$-time algorithm was given \cite{opencup}. The main result that we show for disks is superior to both of these algorithms.

Polishchuk and Mitchell \cite{polishchuk2005touring} showed the case where regions are constrained to be intersections of balls or halfspaces in $d$ dimensions to be a special instance of a second-order cone program (SOCP), which runs in $\BIGO{d^3c^{1.5}n^2\log{\frac{1}{\epsilon}}}$ time using SOCP time bounds as a black box. Here, $c$ is the number of halfspace or ball constraints.

The touring regions problem can be seen as the offline version of chasing convex bodies, in which convex bodies are given sequentially, and the algorithm must choose which point to go to on a convex body before the next convex body is revealed. Bubeck et al. \cite{bubeck-et-al} and Sellke \cite{sellke} showed competitive ratios of $2^{\BIGO{d}}$ and $\BIGO{\sqrt{d \log{n}}}$, respectively. 

\subsection{Formal problem description}

\begin{definition}[Approximate touring regions problem]
Given $n$ sets of points (regions) $R_1, R_2, \dots, R_n$ each a subset of $\R^d$, a starting point $p_0$, and an ending point $p_{n+1}$,\footnote{For convenience, some of our results define the degenerate regions $R_0\triangleq \{p_0\}$ and $R_{n+1}\triangleq \{p_{n+1}\}$. } define the function $D\colon (\R^d)^n \to \R$ as $D(p_1, p_2, \dots, p_n) \triangleq \sum_{i=0}^{n} \vmag{p_i-p_{i+1}}_2$.

\begin{sloppypar}
Let $\mathcal A \triangleq \{(p_1, p_2, \dots, p_n) \mid \forall i, p_i \in R_i\}\subseteq (\R^d)^n$. Find a tuple of points (tour) $(p_1',p_2',\dots,p_n')\in \mathcal A$ such that $D(p_1', p_2', \dots, p_n') \leq (1+\eps)\min_{x \in \mathcal A} D(x)$.
\end{sloppypar}
\end{definition}


We consider three main types of regions: unions of convex bodies, convex fat bodies with constant bounded fatness, and balls.

\begin{definition}[Unions of convex bodies]
\label{definition:union}
We call a region $R \subset \R^d$ a \textbf{union of $k$ convex bodies} if $R = C_1 \cup C_2 \cup \dots \cup C_{k}$ and each $C_i$ is convex and compact. The $C_i$ are allowed to intersect.
\end{definition}

We primarily restrict to the case where $k \leq \BIGO{1}$. 

\begin{definition}[Bounded fatness]
\label{definition:fat}
We say that a convex region $R \subset \R^d$ is \textbf{fat} if there exist balls $h, H$ with radii $0<r_h \le r_H$, respectively, that satisfy $h \subseteq R \subseteq H \subset \R^d$ and $\frac{r_H}{r_h} =\BIGO{1}$. 
\end{definition}

Fat objects have been previously considered in a variety of computational geometry settings \cite{KATZ1997299,10.1016/0925-7721(93)90018-2,Overmars1996RangeSA,OVERMARS1992261}.

One element of the problem that has not yet been determined is how we represent the sets of points $R_1, R_2, \dots, R_n$; this depends on what we restrict the regions to be:
\begin{itemize}
    \item \textbf{Unions of convex bodies:} We are given each region in the form $R_i = C_{i, 1} \cup C_{i, 2} \cup \dots \cup C_{i, k_i}$. Each of the convex bodies $C_{i, j}$ may be accessed via a \emph{closest point oracle.} This oracle allows us to call the function $\closest_{i,j}(p)$ on some point $p$, which returns the point $p' \in C_{i, j}$ such that $\vmag{p-p'}$ is minimized in $\BIGO{1}$ time (note that $p'$ is unique due to convexity).
    
    \item \textbf{Convex fat bodies:} We have access to each of the convex bodies $R_i$ via a closest point oracle. Additionally, for each region, we are given the radius $r_h$ of the inscribed ball (as described in \cref{definition:fat}), and a constant upper bound on the quantity $\frac{r_H}{r_h}$ over all regions. 
    
    \item \textbf{Balls:} For each ball in the input we are given its center  $c \in \R^d$ and its radius $r \in \R_{>0}$.
\end{itemize}

We consider the $2$-dimensional and general $d$-dimensional cases separately. In the $d$-dimensional case, we assume $d$ is a constant (for example, we say $2^d = \BIGO{1}$). We also consider the possibly non-disjoint versus disjoint cases separately, where the latter is defined by the restriction $R_i \cap R_j = \emptyset$ for all $0\le i<j\le n+1$. 

\paragraph*{Motivation for our model}
When considering general convex bodies,  it is natural to augment the model of computation with oracle access to the bodies, including membership, separation, and optimization oracles \cite{goodman}. In fact, when solving the touring regions problem for general convex bodies, a closest point oracle is necessary even for the case of a single region, where the starting point is the same as the ending point and the optimal solution must visit the closest point in the region to the starting point. Closest point oracles can be constructed trivially when the bodies are constant sized polytopes or balls. Closest point oracles have been used in the field of convex optimization \cite{Combettes_Trussell_1990, doi:10.1137/S0036144593251710}.

Our representations for unions of convex bodies, convex fat bodies, and balls, have the nice structure that each ``contains'' the next: we can trivially construct a closest point oracle for balls. Additionally, a ball is a specific type of convex fat body, which in turn is a specific type of convex body. We justify considering convex fat bodies as they are in some sense ``between'' balls and general convex bodies: they obey some of the packing constraints of balls. Considering unions of convex bodies allows us to represent a variety of non-convex and disconnected regions.

\subsection{Summary of results}

Our results and relevant previous results are summarized in \cref{table:results-2d} and \cref{table:results-dd}. We obtain a bound of $\BIGO{n^{2d-1}(\log\log n+\frac{1}{\eps^{2d-2}})}$ time for the most general case of touring unions of convex bodies in $d$ dimensions, where each region is a union of $\BIGO{1}$ convex bodies. This significantly improves to $\BIGO{\frac{n}{\eps^{2d-2}}}$ time if we restrict the regions to be disjoint convex fat bodies. Notice that this latter bound is linear in $n$; in fact, we show that \emph{any} FPTAS for touring convex fat bodies can be transformed into one that is linear in $n$ (\cref{lemma:grouping}). If the regions are further restricted to be balls, we can apply our new technique of placing points nonuniformly, and the time complexity improves to $\BIGO{\frac{n}{\eps^{d-1}}\log^2\frac{1}{\eps}+\frac{1}{\eps^{2d-2}}}$, which roughly halves the exponent of $\frac{1}{\epsilon}$ compared to the convex fat bodies algorithm while retaining an additive $\frac{1}{\epsilon^{2d-2}}$ term. 

Our 2D-specific optimizations allow us to obtain superior time bounds compared to if we substituted $d=2$ into our general dimension algorithms. In most cases, we save a factor of $\frac{1}{\epsilon}$. Notably, for convex fat bodies, we obtain an algorithm with linear time dependence on both $n$ and $\frac{1}{\epsilon}$. For our main result of touring disjoint disks, we combine our optimizations for convex fat bodies and balls with 2D-specific optimizations.

\begin{restatable*}{theorem}{nonIsectDisksTwoD}\label{thm:non-isect-disks}

There is an $\BIGO{\frac{n}{\sqrt{\eps}}\log\frac{1}{\eps}+ \frac{1}{\eps}}$-time algorithm for touring disjoint disks.

\end{restatable*}

With a new \emph{polygonal approximation} technique, we use the result of \cite{starburst-improve} for touring polygons as a black box to obtain algorithms with a square root dependence on $\frac{1}{\epsilon}$, most notably an $\BIGO{\frac{n^{3.5}}{\sqrt{\epsilon}}}$-time algorithm for touring 2D convex bodies and an $\BIGO{\frac{n^2}{\sqrt{\epsilon}}}$-time algorithm for touring 2D disjoint convex fat bodies.

Approximation algorithms for touring polygons in 2D have been well-studied. As mentioned in the introduction, Tan and Jiang \cite{starburst-improve} improved on Dror et al.'s \cite{dror2003} exact solution for convex polygons, while Mozafari and Zarei \cite{Mozafari2012TouringPA} approximated unions of nonconvex polygons, which we include in \cref{table:results-dd} for reference. One of our 2D-specific techniques can be used to improve the result of  \cite{Mozafari2012TouringPA} from $\tilde\MO\paren{\frac{|V|^2n^2}{\eps^2}}$ to $\BIGO{|V|n\log |V|\log \log n+\frac{|V|n}{\eps} \log \frac{|V|}{\eps}+|V|^2\alpha(|V|)}$ time, where $\alpha$ is the inverse Ackermann function, obtaining a strictly better running time for the problem of touring possibly non-disjoint unions of polygons in $2$ dimensions.

The $\BIGO{c^{1.5}n^{2}\log{\frac{1}{\epsilon}}}$-time result for touring $d$ dimensional convex bodies given by \cite{polishchuk2005touring}, where each body is an intersection of balls and half spaces (with a total of $c$ constraints) can be applied specifically to balls to yield an $\BIGO{n^{3.5}\log{\frac{1}{\epsilon}}}$-time algorithm. Our algorithms for touring disjoint disks and balls all take time linear in $n$ and are thus superior when $\epsilon$ is not too small.


\begin{table}[h!]
\centering
\begin{tabular}{||c p{3cm} c p{2.8cm}||} 
 \hline
 Representation & Runtime & Intersecting? & Source \\
 \hline\hline
 Polygon Unions & $\tilde\MO\paren{\frac{|V|^2n^2}{\eps^2}}$ & Yes & Touring Multiple-polygons \cite{Mozafari2012TouringPA} \\
 \hline
 Convex Polygons (Exact) & $\BIGO{|V|n}$, $\BIGO{|V|n^2}$ & No, Yes & Touring Polygons \cite{dror2003}, \cite{starburst-improve} \\
 \hline
 Convex Unions (Oracle Access) & $\BIGO{n^2\paren{\log \log n+\frac{1}{\eps}}}$, $\BIGO{n^3\paren{\log \log n+\frac{1}{\eps}+\frac{\log 1/\eps}{n\eps}}}$ & No, Yes & \textbf{\cref{thm:2d-unions}} \\
 \hline
 Polygon Unions & $\BIGO{\frac{|V|n}{\eps} \log \frac{|V|}{\eps}+\dots}$ & Yes & \textbf{\cref{thm:nonconvex-polygons-2d}} \\
 \hline
 Convex (Oracle Access) & $\BIGO{\frac{n^{2.5}}{\sqrt{\eps}}}$, $\BIGO{\frac{n^{3.5}}{\sqrt{\eps}}}$ & No, Yes & \textbf{\cref{thm:polygonal-approx}} \\
 \hline
 Convex Fat (Oracle Access) & $\BIGO{\frac{n}{\eps}}$, $\BIGO{\frac{n^2}{\sqrt{\eps}}}$ & No &  \textbf{Theorems \ref{thm:fat-non-isect-2d-1}, \ref{thm:fat-non-isect-2d-2}} \\
 \hline
 Disks & $\BIGO{\frac{n}{\sqrt{\eps}}\log\frac{1}{\eps}+ \frac{1}{\eps}}$ & No & \textbf{\cref{thm:non-isect-disks}} \\
 \hline
\end{tabular}
\caption{Previous and new bounds on touring $n$ regions in two dimensions up to multiplicative error $1+\eps$, where $\eps\le \BIGO{1}$. For polygons, $|V|$ is the total number of vertices over all polygons. For unions of convex bodies, we assume that every region is composed of $\BIGO{1}$ bodies.} 
\label{table:results-2d}
\end{table}

\begin{table}[h!]
\centering
\begin{tabular}{||p{4cm} c c c||} 
 \hline
 Representation & Runtime & Intersecting? & Source \\
 \hline\hline
 Convex Bodies, each an intersection of balls or halfspaces & $\BIGO{c^{1.5}n^2\log \frac{1}{\eps}}$ & Yes & SOCP \cite{polishchuk2005touring} \\
 \hline
 Convex Unions (Oracle Access) & $\BIGO{n^{2d-1}(\log\log n+\frac{1}{\eps^{2d-2}})}$ & Yes & \textbf{\cref{thm:d-dim}} \\
 \hline
 Convex Fat (Oracle Access) & $\BIGO{\frac{n}{\eps^{2d-2}}}$ & No & \textbf{\cref{thm:fat-non-isect-d}} \\
 \hline
 Balls & $\BIGO{\frac{n}{\eps^{d-1}}\log^2\frac{1}{\eps}+\frac{1}{\eps^{2d-2}}}$ & No & \textbf{\cref{thm:non-isect-balls}} \\
 \hline
\end{tabular}
\caption{Previous and new bounds on touring $n$ regions in $d\ge 2$ dimensions up to multiplicative error $1+\eps$, where $\eps\le \BIGO{1}$. Note that $d$ is treated as a constant. For polyhedra, $c$ is the total number of constraints. For unions of convex bodies, we assume that every region is composed of $\BIGO{1}$ bodies.}
\label{table:results-dd}
\end{table}

\subsection{Organization of the paper}

We start in \cref{section:convex-bodies} by considering unions of general convex bodies, using the closest point projection, pseudo-approximation, and 2D-specific optimizations. We then use the ideas of packing and grouping to obtain algorithms for convex fat bodies in \cref{section:fat-bodies}. Finally, we optimize specifically for balls in \cref{section:balls} by placing points non-uniformly.

\subsection{Summary of techniques}\label{sec:summary-techniques}

Here, we introduce the techniques mentioned in the previous subsection.

\paragraph*{Placing points uniformly (\cref{section:convex-bodies})}

A general idea that we use in our approximation algorithms is to approximate a convex body well using a set of points on its boundary. For previous results involving polygons or polyhedra \cite{asano2004pseudo,Mozafari2012TouringPA}, this step of the process was trivial, as points were equally spaced along edges. In order to generalize to convex bodies in arbitrary dimensions, we equally space points on boundaries using the closest point projection oracle with a bounding hypercube (\cref{lemma:equal-spacing-points}). 

After discretizing each body into a set of points, we can solve the problem in polynomial time using dynamic programming (DP): for each point, we find and store the optimal path ending at it by considering transitions from all points on the previous region.


\paragraph*{Pseudo-Approximation (\cref{section:convex-bodies})}

Let $OPT$ be the optimal path length for touring regions. Consider some guess of the optimal path length for touring convex bodies $L_{APPROX} \geq OPT$, and then consider constructing a hypercube of side length $2L_{APPROX}$ centered at the starting point. We then equally space points on the boundary of the portions of the convex bodies which are inside the hypercube, and solve the problem assuming that the optimal path must travel through these points, which adds some additive error proportional to $L_{APPROX}$ to the optimal path length. This is called a \emph{pseudo-approximation} because the error bound only holds if $L_{APPROX} \geq OPT$, and if $L_{APPROX}$ is much bigger than $OPT$, the additive error is very large.

The idea for using this pseudo-approximation to compute an actual approximation is to start with some naive large approximation of the optimal path length, and then continuously guess smaller values of the optimal path length and run the pseudo-approximation as a subroutine to generate better approximations, eventually finding a constant approximation in $\log{\log{n}}$ iterations. Once a constant approximation is found, the pseudo-approximation algorithm becomes an actual approximation algorithm, and is then used to find a $(1+\epsilon)$-approximation. This method was used previously by \cite{asano2004pseudo} and \cite{Mozafari2012TouringPA}. An exposition about pseudo-approximation can be found in \cite{asano2004pseudo}.

\paragraph*{2D-specific optimizations (\cref{section:convex-bodies})}

Previous approximation algorithms for related problems discretize the boundary of each convex region using $\BIGO{\frac{1}{\epsilon}}$ points. We present a new approach to approximate each boundary using a convex polygon with $\BIGO{\frac{1}{\sqrt{\epsilon}}}$ vertices (\cref{lemma:half-eps-intersect}). This allows us to use previous exact algorithms for touring convex polygons as black boxes. 

A separate approach is to use additively weighted Voronoi diagrams (\cref{lemma:additive-voronoi}) to optimize dynamic programming (DP) transitions from quadratic to near-linear time. When we additionally assume the input shapes are disjoint, we use properties of Monge matrices to optimize the transitions to expected linear time (\cref{lemma:additive-linear}).

\paragraph*{Packing and grouping (\cref{section:fat-bodies})}
While our general algorithm for unions of convex regions has runtime $\tilde \MO \left (\frac{n^{2d-1}}{\eps^{2d-2}} \right )$, we are able to improve this to $\BIGO{\frac{n}{\eps^{2d-2}}}$ time for convex fat bodies. The key ideas behind this improvement are \emph{packing} and \emph{grouping}.

We use a simple \emph{packing} argument to show that the path length for visiting $n$ disjoint convex fat bodies with radius $r$ must have length at least $\Omega(r \cdot n)$ for sufficiently large $n$ (\cref{lemma:packing-lb}). This was used by \cite{opencup} for the case of unit disks. However, it is not immediately clear how to use this observation to obtain improved time bounds when convex fat regions are not all restricted to be the same size.

The idea of \emph{grouping} is to split the sequence of regions into smaller contiguous subsequences of regions (groups). In each group, we find the minimum-sized region, called a \emph{representative} region, which allows us to break up the global path into smaller subpaths between consecutive representatives. The earlier packing argument now becomes relevant here, as we can show a lower bound on the total length of the optimal path in terms of the sizes of the representatives.

\paragraph*{Placing points non-uniformly (\cref{section:balls})}

Previous approximation methods rely on discretizing the surfaces of bodies into evenly spaced points. For balls, we use the intuition that the portion of the optimal path from one ball to the next is ``long'' if the optimal path does not visit the parts of the surfaces that are closest together. This allows us to place points at a lower density on most of the surface area of each ball, leading to improved time bounds. We use this technique in conjunction with packing and grouping. For disks, we additionally apply the previously mentioned 2D-specific optimizations.

\section{Convex bodies}
\label{section:convex-bodies}

First, we consider the most general case of convex bodies (or unions of convex bodies), as variations of these techniques also apply to later results. We split the discussion into the general $d$-dimensional case and the $2$-dimensional case. Omitted proofs for this section may be found in \cref{sec:omitted-proofs-convex-bodies}.

\subsection{General dimensions}

For the theorems in this section, we let $|R_i|$ denote the number of convex bodies that together union to region $R_i$. Recall that the convex bodies which make up region $R_i$ may overlap.

The first main ingredient is the closest point projection, which allows us to equally space points on each convex body. The proof is deferred to the appendix.

\begin{lemma}\label{lemma:project-contract}
For a convex region $C$, define $\text{closest}_C(p)\triangleq \text{argmin}_{c\in C}\vmag{c-p}$. For any two points $p_1$ and $p_2$, $\vmag{\text{closest}_C(p_1)-\text{closest}_C(p_2)}\le \vmag{p_1-p_2}$.
\end{lemma}

For any closed set $X$, let $\partial X$ denote the boundary of $X$.

\begin{lemma}[Equal spacing via closest point projection]
\label{lemma:equal-spacing-points}
Given a convex body $C$ for which we have a closest point oracle and a hypercube $\mathcal H$ with side length $r$, we can construct a set $S\subset C$ of $\BIGO{\frac{1}{\eps^{d-1}}}$ points such that for all $p \in (\partial C) \cap \mathcal H$, there exists $p' \in S$ such that $\vmag{p-p'} \leq r \eps$.
\end{lemma}

\begin{proof}[Proof Sketch]
First, we prove the statement for $C = \mathcal H$. For this case, it suffices to equally space points on each face of an axis-aligned hypercube defined by $[0, r]^d$. For example, for the face defined by $x_d = 0$, we place points in a lattice at all coordinates $(x_1, x_2, \dots, x_{d-1}, x_d)$ that satisfy $x_d = 0$ and $x_i = k_i \cdot r\eps$ for all integers $k_i\in \left[0, \frac{1}{\eps}\right]$. For $C \neq \mathcal H$, equally space points on $\mathcal H$ as we stated to create a set $S_{\mathcal H}$. Then define $S\triangleq\{\closest_C(s) \mid s \in S_{\mathcal H}\}$. The proof that $S$ satisfies the conditions of the lemma is deferred to the appendix.
\end{proof}

Now, we introduce the concept of the pseudo-approximation, which takes in an accuracy parameter $\gamma$ and an estimate of the optimal path length $L_{APPROX} \geq OPT$ and reduces each region $R_i$ to a finite set of points $S_i \subset R_i$ such that the optimal tour for touring $S_i$ is also a tour for $R_i$, and has length at most $OPT + \gamma L_{APPROX}$. 

Note that when the regions are possibly non-disjoint, it is not true that the optimal path must visit each $\partial R_i$, so more care must be taken. In particular, we use the fact that the only time an optimal path does not visit the boundary of $R_i$ is when it visited the boundary of some region $\partial R_l$ for some $l < i$ and then remained on the interior of regions $R_{l+1}, R_{l+2}, \dots, R_i$, in which case the optimal path has moved $0$ distance when visiting these regions. This requires more effort to bound the error from the optimal length and makes the dynamic programming transitions more complex, but both algorithms achieve the same time bound when all $|R_i| \leq \mathcal O(1)$ (that is, when all regions are unions of a constant number of convex bodies). 
\begin{lemma}[$\gamma$ pseudo-approximation]\label{lemma:pseudo-approx}
\begin{sloppypar}
Given an estimate of the optimal path length $L_{APPROX}$ and $0 < \gamma \leq 1$, if $OPT \leq L_{APPROX}$, we can construct a valid solution with length at most $OPT + \gamma L_{APPROX}$. If all $R_i$ are disjoint, this construction takes $\BIGO{\paren{\frac{n}{\gamma}}^{2d-2}\sum_{i=1}^{n-1} |R_i||R_{i+1}|}$ time. When the $R_i$ are possibly non-disjoint, the runtime increases to $\BIGO{\paren{\frac{n}{\gamma}}^{2d-2}\max_j|R_{j}|\sum_{i=1}^{n} |R_i|}$.
\end{sloppypar}
\end{lemma}

\begin{proof}[Proof Sketch]
We construct finite sets of points $S_i \subset R_i$ such that there exists a path of the desired length that tours regions $S_i$. 

Assume $OPT\le L_{APPROX}$ and consider a hypercube $\mathcal H$ centered at $p_0$ with side length $4 \cdot L_{APPROX}$. We define the construction as follows: For each of the convex bodies \\$C_{i, 1}, C_{i, 2}, C_{i, 3}, \dots, C_{i, |R_i|}$ that make up $R_i$, apply the construction given in \cref{lemma:equal-spacing-points} with $\epsilon \triangleq \frac{\gamma}{16n}, r \triangleq 4L_{APPROX}$, and set $S_i$ to be the union of all $|R_i|$ constructed sets of points. This uses $\BIGO{\frac{1}{\epsilon^{d-1}}} = \BIGO{\paren{\frac{n}{\gamma}}^{d-1}}$ points per convex body. 

Given the sets $S_i$, the path of length $OPT + \gamma L_{APPROX}$ can be computed directly using dynamic programming. When the $R_i$ are disjoint, transitions occur from points on $S_i$ to $S_{i+1}$. Transitions are slightly different for the possibly non-disjoint case since points in the set $S_i$ can transition to points on $S_j$ for $j > i+1$. The details of the dynamic programming and the proof that the length of the returned path is bounded above by $OPT+\gamma L_{APPROX}$ are deferred to the appendix.
\end{proof}




To convert our pseudo-approximation algorithm into an actual approximation, we start with an $\BIGO{n}$-approximation of the optimal path length; the construction is deferred to the appendix.

\begin{lemma}
\label{lemma:trivial-n-approx}
There is a trivial $(2n+1)$-approximation for touring general regions that can be computed in $\BIGO{n}$ time given a closest point oracle.
\end{lemma}

Now, our goal is to construct a constant approximation starting from our trivial approximation. The idea, first presented in \cite{asano2004pseudo} for the problem of 3D shortest path with obstacles, is to run the pseudo-approximation $\BIGO{\log \log n}$ times with $\gamma=1$. The resulting runtime is much faster than if one were to naively apply \cref{lemma:pseudo-approx} with $\gamma = \frac{\epsilon}{n}$. The proof is deferred to the appendix.

\begin{lemma}[Constant approximation via pseudo-approximation]\label{lemma:non-intersect-d-dim-constant-approx}
There is an \\
$\BIGO{n^{2d-2}\log{\log{n}} \cdot (\sum_{i=1}^{n-1} |R_i| |R_{i+1}|)}$-time algorithm that obtains a $4$-approximation for touring disjoint unions of convex bodies in $d$ dimensions. If the unions can intersect, the runtime increases to $\BIGO{n^{2d-2}\log{\log{n}} \cdot (\sum_{i=1}^{n} |R_i| \max_{j} |R_j|)}$.
\end{lemma}

Finally, we combine all of the lemmas of the section to give the main results.

\begin{theorem}\label{thm:d-dim}
There is an $\BIGO{n^{2d-2}(\log{\log{n}}+\frac{1}{\eps^{2d-2}}) \cdot (\sum_{i=1}^{n-1} |R_i| |R_{i+1}|)}$-time algorithm for touring disjoint unions of convex bodies in $d$ dimensions. When the bodies are allowed to intersect, the runtime becomes $\MO\Big (n^{2d-2}(\log{\log{n}}+\frac{1}{\eps^{2d-2}}) \cdot (\sum_{i=1}^{n} |R_i| \max_j |R_j|)\Big)$.
\end{theorem}

\begin{proof}
\begin{sloppypar}
For the disjoint case, apply \cref{lemma:non-intersect-d-dim-constant-approx} to get a constant approximation in $\MO\big (n^{2d-2}\log{\log{n}} \cdot (\sum_{i=1}^{n-1} |R_i| |R_{i+1}|)\big)$ time, then use \cref{lemma:pseudo-approx} with $L_{APPROX}$ as our constant approximation and $\gamma = \frac{\epsilon}{4}$ to obtain a $(1+\eps)$-approximation in $\BIGO{n^{2d-2}(\frac{1}{\eps^{2d-2}}) \cdot (\sum_{i=1}^{n-1} |R_i| |R_{i+1}|)}$ additional time. The possibly non-disjoint case is similar.\qedhere
\end{sloppypar}
\end{proof}

\subsection{Two dimensions}

When the unions of convex bodies are constrained to lie in 2D, there are two main avenues for further improvements: first, by speeding up the dynamic programming (DP) transitions when all regions have been discretized into point sets, and second, by approximating convex bodies by convex polygons instead of sets of points. In this section, ``union of convex bodies'' refers to a union of $\BIGO{1}$ convex bodies per region.

\subsubsection{Dynamic programming speedup}

The first speedup comes as a result of observing that the DP in \cref{lemma:pseudo-approx} is similar to closest point queries, which can be computed efficiently.

\begin{lemma}[Additive Voronoi]\label{lemma:additive-voronoi}
\begin{sloppypar}
Given two lists of points $B=[b_1,\dots,b_m]$ and $A=[a_1,a_2,\dots,a_n]$ and a real weight $[w_1,\dots,w_n]$ for each point in $A$, we can compute $\min_{1\le j\le n}\paren{w_j+\vmag{a_j-b_i}}$ for each $i\in [1,m]$ in $\BIGO{(m+n)\log n}$ time.
\end{sloppypar}
\end{lemma}

\begin{proof}
This problem is equivalent to constructing and querying a Voronoi diagram for additively weighted point sets. Constructing the diagram can be done in $\BIGO{n\log n}$ time by a variant of Fortune's algorithm \cite{fortune1987sweepline}. For each $b_i$ we can search the Voronoi diagram for the $a_j$ corresponding to the minimum in $\BIGO{\log n}$ time.
\end{proof}

\begin{corollary}\label{corollary:additive-voronoi-app}
The Touring Regions Problem in 2D, where all $R_i$ are sets of finitely many points $S_i$, can be solved exactly in $\BIGO{\sum_{i=1}^{n}|S_i|\log{|S_i|}}$ time. 

\end{corollary}

\begin{proof}
Recall the dynamic programming method from \cref{lemma:pseudo-approx}, which computes the DP value for each point in $S_{i+1}$ in $\BIGO{|S_i|}$ time, meaning that each pair of adjacent regions contributes $|S_i||S_{i+1}|$ to the runtime. Substituting \cref{lemma:additive-voronoi} in place of this step, the runtime improves to
$\BIGO{\sum_{i=1}^{n-1}(|S_i|+|S_{i+1}|)\log |S_i|}\le \BIGO{\sum_{i=1}^n|S_i|\log |S_i|}$. 
\end{proof}

For disjoint convex regions, we use a stronger guarantee than \cref{lemma:additive-voronoi}:

\begin{lemma}\label{lemma:additive-linear}
Given are the vertices of two disjoint convex polygons $B=[b_1,\dots,b_m]$ and $A=[a_1,a_2,\dots,a_n]$ in counterclockwise order and real weights $[w_1,\dots,w_n]$, one for each vertex of $A$. Define $d(i,j)\triangleq w_j+\vmag{a_j-b_i}$. Then  $\min_{1\le j\le n}d(i,j)$ may be computed for all $i\in [1,m]$ in $\BIGO{m+n}$ expected time.
\end{lemma}

\begin{proof}
We first discuss the case where all $w_i=0$. Aggarwal and Klawe \cite{aggarwal1990applications} showed how to reduce the computation of $\min_{1\le j\le n, a_j\text{ visible from }b_i}d(i,j)$ and 
$\min_{1\le j\le n, a_j\text{ not visible from }b_i}d(i,j)$ for all $i\in [1,m]$ to computing the row minima of several \emph{Monge partial matrices} with dimensions $m_1\times n_1, m_2\times n_2,\dots, m_k\times n_k$ such that $\sum (m_i+n_i)\le \MO(m+n)$ in $\MO(m+n)$ time. Here, $a_j$ is said to be visible from $b_i$ if the segment $\overline{a_jb_i}$ intersects neither the interiors of polygons $A$ nor $B$. The definition of Monge partial matrix can be found in \cite{chan2021near}. 

Chan \cite{chan2021near} recently introduced an $\BIGO{m+n}$ expected time randomized algorithm for computing the row minima of an $m\times n$ Monge partial matrix.\footnote{The Monge partial matrix does not have to be given explicitly; it suffices to provide an oracle that returns the value of any entry of the matrix in $\mathcal O(1)$ time.} Thus, the case of $w_i=0$ can be solved in $\BIGO{m+n}$ expected time. 

The key claim that Aggarwal and Klawe \cite{aggarwal1990applications} use to show that all the matrices they construct are Monge partial is as follows:

\begin{claim}[Lemma 2.1 of \cite{aggarwal1990applications}, adapted]
Assume all $w_j=0$. Suppose $j\neq j'$ and $i\neq i'$. If $a_{j}a_{j'}b_{i'}b_i$ form a convex quadrilateral in that order then $d(i,j)+d(i',j')\le d(i,j')+d(i',j)$.
\end{claim}

The claim above holds by the triangle inequality, and it is easy to check that it still holds without the assumption $w_j=0$. Thus the algorithm from \cite{aggarwal1990applications} generalizes to the case of nonzero $w_j$ with minor modifications.
\end{proof}

\begin{corollary}\label{corollary:linear-touring}
The Touring Regions Problem in 2D, where all $R_i$ are sets of finitely many points $S_i$ that each form a convex polygon in counterclockwise order and the convex hulls of all $S_i$ are disjoint, can be solved exactly in $\BIGO{\sum_{i=1}^{n}|S_i|}$ expected time.
\end{corollary}

Using these techniques, we obtain the following speedups. \cref{thm:nonconvex-polygons-2d} follows due to similar reasoning as \cref{thm:2d-unions}, as a polygon with $|V|$ vertices can be triangulated in $\BIGO{|V|}$ time due to Chazelle \cite{triangulate-linear}.

\begin{theorem}\label{thm:2d-unions}
There is an $\BIGO{n^2\paren{\log \log n+\frac{1}{\eps}}}$-time algorithm for touring disjoint unions of convex bodies in 2D where each union consists of $\BIGO{1}$ convex bodies. When the bodies are possibly non-disjoint, the bound is $\BIGO{n^3\paren{\log \log n+\frac{1}{\eps}+\frac{\log 1/\eps}{n\eps}}}$ time.
\end{theorem}

\begin{proof}
For the first bound, use \cref{thm:d-dim} with \cref{corollary:linear-touring} to speed up DP transitions. For the second bound, use \cref{thm:d-dim} but with an extension of \cref{corollary:additive-voronoi-app} to speed up DP transitions. 
\end{proof}

\begin{theorem}\label{thm:nonconvex-polygons-2d}
There is an $\BIGO{|V|n\log |V|\log \log n+\frac{|V|n}{\eps} \log \frac{|V|}{\eps}}$-time algorithm for touring disjoint unions of polygons. When the polygons are allowed to intersect each other, the time complexity increases by $\BIGO{|V|^2\alpha(|V|)}$.
\end{theorem}

\subsubsection{Polygonal approximation algorithms}

Up until now, we have approximated the perimeter of a convex region using points. We can alternatively approximate the perimeter using a convex polygon with fewer vertices, which can be computed using our closest point projection oracle. The proof is deferred to the appendix.

\begin{lemma}[Polygonal approximation]\label{lemma:half-eps-intersect}
Given a closest point oracle for a convex region $C$ and a unit square $U$, we may select $\BIGO{\eps^{-1/2}}$ points in $C$ such that every point within $C\cap U$ is within distance $\eps$ of the convex hull of the selected points.
\end{lemma}

The polygonal approximation allows us to immediately obtain the following result. The proof is deferred to the appendix.

\begin{theorem}
\label{thm:polygonal-approx}
There is a $\MO \paren{\frac{n^{2.5}}{\sqrt{\eps}}}$-time algorithm for touring disjoint convex bodies in 2D. When the convex bodies are possibly non-disjoint, the bound is $\BIGO{\frac{n^{3.5}}{\sqrt{\eps}}}$ time.
\end{theorem}

\begin{proof}
Let's start with the disjoint case.
We first use \cref{thm:2d-unions} with $1+\epsilon = 2$ to obtain a path of length $L_{APPROX}$ that satisfies $L_{APPROX} \leq 2\cdot OPT$ in $o(n^{2.5})$ time.

Consider constructing a square $\mathcal H$ of side length $2L_{APPROX}$ centered at $p_0$. Let $\epsilon' = \frac{\epsilon}{8n}$. 
Now, we apply \cref{lemma:half-eps-intersect} to select a set $S_i$ of size $|S_i| \leq \BIGO{\epsilon'^{-1/2}}$ points on each region $R_i$, such that every point within $R_i \cap \mathcal H$ is within distance  $\epsilon' \cdot L_{APPROX}$ of some point in the convex hull of $S_i$.  

Define $C_i$ to be the convex hull of $S_i$ in counterclockwise order, which we can compute in linear time because the construction given in \cref{lemma:half-eps-intersect} returns points that are all on the convex hull. Now, run \cite{starburst-improve} to solve the touring disjoint convex polygons problem for $C_i$ exactly in $\BIGO{|V|n} = \BIGO{n^2 \epsilon'^{-1/2}} \leq \BIGO{\frac{n^{2.5}}{\sqrt{\epsilon}}}$ time. Recall that $|V|$ is the total number of vertices over all polygons.

It remains to show that the solution we find from the convex polygons problem is a $1+\epsilon$ approximation of the answer. Consider an optimal solution $p_0 \in R_0, p_1 \in R_1, \dots, p_n \in R_n, p_{n+1} \in R_{n+1}$. Now, for every $i$, define $p_i'$ to be the closest point on $C_i$ to $p_i$, where \cref{lemma:half-eps-intersect} guarantees $\vmag{p_i-p_i'} \leq \epsilon' \cdot L_{APPROX} \leq 2\epsilon' \cdot OPT$. Thus, 
\begin{align*}
    \sum_{i=0}^{n} \vmag{p_i'-p_{i+1}'} &\leq \sum_{i=0}^{n} \paren{\vmag{p_i-p_{i+1}} + \vmag{p_i-p_i'} + \vmag{p_{i+1}-p_{i+1}'}} \\
    &\leq OPT + (n+1) \cdot 4\epsilon' \cdot OPT \leq (1+\epsilon)OPT,
\end{align*}
as desired.

For the intersecting case, we first use \cref{thm:2d-unions} with $1+\epsilon = 2$ to obtain a constant approximation of the optimal length in $o(n^{3.5})$ time. The rest of the proof is identical, except now the $C_i$ can intersect, which changes the runtime of the application of \cite{starburst-improve} to $\BIGO{|V|n^2} = \BIGO{n^3 \epsilon'^{-1/2}} \leq \BIGO{\frac{n^{3.5}}{\sqrt{\epsilon}}}$.
\end{proof}

\section{Disjoint convex fat bodies}
\label{section:fat-bodies}
In this section, we present packing and grouping techniques for touring disjoint convex fat bodies and show how they can be applied to obtain $\BIGO{\min\paren{\frac{n}{\eps},\frac{n^2}{\sqrt \eps}}}$-time algorithms for touring disjoint convex fat bodies in $2$ dimensions. Omitted proofs for this section may be found in \cref{sec:omitted-proofs-convex-fat}.

\subsection{Techniques}

\subsubsection{Packing}

A packing argument shows that the length of the optimal path length is at least linear in the number of bodies and the minimum $r_h$ (that is, the minimum radius of any inscribed ball). Intuitively, if we place $n$ disjoint objects of radius at least $1$ that are close to being disks on the plane, the length of the optimal tour that visits all of them should be at least linear in $n$ for sufficiently large $n$. The details are in the appendix.

\begin{lemma}[Packing Lemma]\label{lemma:packing-lb}
Assume a fixed upper bound on $\frac{r_H}{r_h}$. Then there exists $n_0= \BIGO{1}$ such that the optimal path length $OPT$ for touring any $n \ge n_0$ disjoint convex fat objects is $\Omega(n \cdot \min r_h)$. For balls, $n_0=3$.
\end{lemma}

The packing lemma allows us to obtain a strong lower bound on the length of the optimal tour in terms of the size of the regions, which will be crucial in proving that our algorithms have low relative error. 

\begin{corollary}\label{corollary:sp-lower-bound}

Let $r_i$ denote the $i$th largest $r_h$. For all $i\ge n_0$, $r_i\le \BIGO{\frac{OPT}{i}}$.

\end{corollary}

\begin{proof}
Consider dropping all regions except those with the $i$ largest inner radii and let $OPT_i$ be the optimal length of a tour that visits the remaining disks in the original order. By \cref{lemma:packing-lb}, for $i\ge n_0$, $OPT\ge OPT_i\ge \Omega(i\cdot r_i)\implies r_i\le \BIGO{\frac{OPT}{i}}$.
\end{proof}

\begin{lemma}\label{lemma:sp-logn}
\begin{sloppypar}
The optimal path length for touring $n$ disjoint convex fat bodies is $\Omega \Big( \sum_{i\ge n_0} r_i/\log{n} \Big)$, and there exists a construction for which this bound is tight.
\end{sloppypar}
\end{lemma}

\begin{proof}[Proof Sketch]
Using \cref{corollary:sp-lower-bound},
\begin{equation*}
\frac{\sum_{i\ge n_0}r_i}{\log n}\le \sum_{i\ge n_0}\frac{\BIGO{\frac{OPT}{i}}}{\log n}\le \BIGO{\frac{OPT}{\log n}\sum_{i=n_0}^n\frac{1}{i}}\le \BIGO{OPT}.
\end{equation*}
We display the construction in \cref{fig:sp-logn}; we defer the full description to the appendix. The idea is to place disjoint disks of radii $1/1, 1/2, 1/3, \dots$ such that they are all tangent to a segment of the $x$-axis of length $\BIGO{1}$.
\end{proof}

\begin{figure}[h]
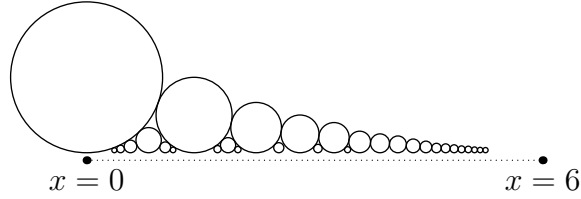

    \centering
    \begin{asy}
    unitsize(1cm);
    draw(circle((0,1.0),1.0));
    draw(circle((1.4142135623730951,0.5),0.5));
    draw(circle((2.230710143300821,0.3333333333333333),0.3333333333333333));
    draw(circle((2.808060412490447,0.25),0.25));
    draw(circle((3.255274007990405,0.2),0.2));
    draw(circle((0.816496580927726,0.16666666666666666),0.16666666666666666));
    draw(circle((3.5933357098818113,0.14285714285714285),0.14285714285714285));
    draw(circle((3.8605969517942356,0.125),0.125));
    draw(circle((4.096299212189751,0.1111111111111111),0.1111111111111111));
    draw(circle((1.861427157873053,0.1),0.1));
    draw(circle((4.297306775241593,0.09090909090909091),0.09090909090909091));
    draw(circle((0.5773502691896257,0.08333333333333333),0.08333333333333333));
    draw(circle((4.464555177243011,0.07692307692307693),0.07692307692307693));
    draw(circle((1.0347144711637184,0.07142857142857142),0.07142857142857142));
    draw(circle((2.528852540300793,0.06666666666666667),0.06666666666666667));
    draw(circle((4.603230226299319,0.0625),0.0625));
    draw(circle((4.724498038817485,0.058823529411764705),0.058823529411764705));
    draw(circle((3.0437626728859626,0.05555555555555555),0.05555555555555555));
    draw(circle((4.835781015632416,0.05263157894736842),0.05263157894736842));
    draw(circle((0.4472135954999579,0.05),0.05));
    draw(circle((1.722820262297279,0.047619047619047616),0.047619047619047616));
    draw(circle((4.93360421324132,0.045454545454545456),0.045454545454545456));
    draw(circle((5.022515058136197,0.043478260869565216),0.043478260869565216));
    draw(circle((1.9905266027466335,0.041666666666666664),0.041666666666666664));
    draw(circle((3.434159446190388,0.04),0.04));
    draw(circle((5.104301140147149,0.038461538461538464),0.038461538461538464));
    draw(circle((1.1375833711384464,0.037037037037037035),0.037037037037037035));
    draw(circle((5.178426071813259,0.03571428571428571),0.03571428571428571));
    draw(circle((5.248612312447619,0.034482758620689655),0.034482758620689655));
    draw(circle((0.3651483716701107,0.03333333333333333),0.03333333333333333));
    
    draw((0,-0.1)--(6,-0.1),dotted);
    dot("$x=0$",(0,-0.1),S);
    dot("$x=6$",(6,-0.1),S);
    
    \end{asy}
    \caption{Construction from \Cref{lemma:sp-logn}: placement of the first $30$ disks}
    \label{fig:sp-logn}
\end{figure}

\subsubsection{Grouping}

We now show that we can split up the optimal path into smaller subpaths by splitting the sequence of bodies into groups of consecutive bodies, finding the minimum-sized body in each group, and considering the subpaths between these small bodies. By the packing lemma, the sum of the radii of the representatives is small compared to the total path length.

In particular, using groups of size $\frac{1}{\epsilon}$, we can compress the smallest sized region into a single point, meaning that we can consider touring regions between these points independently from each other. This allows us to turn \emph{any} polynomial time approximation scheme for touring disjoint convex fat bodies into one that is linear in $n$. 

\begin{lemma}[Grouping Lemma]
\label{lemma:grouping}
Given an algorithm for touring disjoint convex fat bodies in $d$ dimensions that runs in $f(n, \eps)$ time, where $f$ is a polynomial, we can construct an algorithm that runs in $\MO\left (n\eps+1\right) \cdot f\left(\frac{1}{\eps}, \eps\right)$ time (for $\eps \leq \MO(1)$).
\end{lemma}

\begin{proof}
We describe an algorithm achieving a $(1+\MO(\eps))$-approximation. To achieve a $(1+\eps)$-approximation, scale down $\eps$ by the appropriate factor.
 
Define $s \triangleq \left \lceil \frac{1}{\epsilon} \right \rceil$ and let $n_0$ be the constant defined in the statement of \cref{lemma:packing-lb}. We will prove the statement for all $\epsilon$ satisfying $\frac{1}{\epsilon} \geq n_0$. First, we divide the $n+2$ regions (including $R_0$ and $R_{n+1}$) into $k=\max\paren{\ceil{\frac {n+2}{s}},2}\leq\BIGO{n\eps+1}$ consecutive subsequences, each with exactly $s$ regions (except the starting and ending subsequences, which are allowed to have fewer).

Let $M_i$ be the region with minimum inscribed radius $r_h$ in the $i$th subsequence; note that $M_1=R_0$ and $M_k=R_{n+1}$. For each $i \in[1,k]$, pick an arbitrary point $p_i\in M_i$. Let $OPT'$ be the length of the shortest tour of $R_0,\dots,R_{n+1}$ that passes through all of the $p_i$. The $p_1,\dots,p_k$ form $k-1$ subproblems, each with at most $2s$ regions. Therefore, we can $(1+\eps)$-approximate $OPT'$ by $(1+\eps)$-approximating each subproblem in
$(k-1) \cdot f(2s,\epsilon)\leq \MO\left (n\eps+1\right) \cdot f\left(\frac{1}{\eps}, \eps\right)$ time.

It remains to show that $OPT'$ is a $(1+O(\epsilon))$-approximation for $OPT$. Let $r_i$ be shorthand for the radius $r_h$ of $M_i$ ($r_1=r_k=0$). By the definition of fatness, the distance between any two points in $M_i$ is at most $\BIGO{r_i}$. By following through $OPT$ and detouring to each point $p_i$, we get a path through points $p_i$ with length at most $OPT+\BIGO{\sum r_i}$, and $OPT'$ is at most this amount. 

The last remaining step is to show $\sum r_i\leq \BIGO{\epsilon \cdot OPT}$. We apply \cref{lemma:packing-lb} to each subsequence, and obtain that $r_is\leq \BIGO{OPT_i}$, where $OPT_i$ is the optimal distance to tour regions in subsequence $i$. Note that although the starting and ending subsequences can have sizes less than $s$, they satisfy $r_i=0$, so this bound holds for all subsequences. Therefore, $\sum r_i\leq \BIGO{\epsilon\cdot \sum OPT_i}\leq \BIGO{\epsilon\cdot OPT}$.
\end{proof}

\subsection{Algorithms for convex fat bodies}

Using a similar grouping argument, but using constant sized instead of $\frac{1}{\epsilon}$ sized groups, along with earlier methods of using estimates of the path length to place points on the boundaries of the convex fat bodies yields the following results.

\begin{theorem}\label{thm:fat-non-isect-d}
There is an $\BIGO{\frac{n}{\eps^{2d-2}}}$-time algorithm for touring disjoint convex fat bodies in $d$ dimensions.
\end{theorem}

\begin{proof}
We proceed in a similar fashion as \cref{lemma:grouping}, except we define $s \triangleq n_0$, i.e., using constant sized groups instead of $\lceil \frac{1}{\epsilon} \rceil$ sized groups. Let the $M_i$ be defined as in the proof of \cref{lemma:grouping}, and define $m_i$ to be the outer radius of $M_i$.

For each pair of regions $M_i,M_{i+1}$, pick arbitrary points $a\in M_i,b\in M_{i+1}$, and use \cref{thm:d-dim} to
obtain a $4$-approximation $D_{approx}$ of the length of the shortest path from $a$ to $b$ in $\BIGO{1}$ time.
Suppose that the optimal path uses $p\in M_i, q\in M_{i+1}$ and the shortest path from $a$ to $b$ has distance $OPT_{a,b}$;
by the triangle inequality, we must have
\[
    \frac 14 D_{approx} \leq OPT_{a,b} \leq OPT_i + 2m_i + 2m_{i+1}.
\]

Now, consider the path where we start at $p$ and then travel along the line segment from $p$ to $a$, the approximate path of length $D_{approx}$ from $a$ to $b$ (visiting the regions in between $M_i$ and $M_{i+1}$), and the line segment from $b$ to $q$. This path has length at most $D_{approx}+2m_i+2m_{i+1}$, and upper bounds the length of the optimal path between $p$ and $q$. So, the entire path between $p$ and $q$ lies within a ball of radius $D_{approx}+4m_i+2m_{i+1}$ centered at $a$; call this ball $L$. Note that $L$ has radius $l = D_{approx}+4m_i+2m_{i+1} \leq \BIGO{OPT_i+m_i+m_{i+1}}$.

For each region $R_j$ between $M_i$ and $M_{i+1}$ inclusive, we apply \cref{lemma:equal-spacing-points} with the region and a hypercube containing
$L$, which has side length $2l$. Note that points are placed twice on each $M_i$; this is fine. \cref{lemma:equal-spacing-points} guarantees
the existence of a point in $R_j$ that is $2l\epsilon$ close to the point $OPT$ uses by placing $\BIGO{\frac 1{\epsilon^{d-1}}}$ points on each region.

We now bound the difference between the optimal and the shortest paths using only the points we placed.
The difference is at most
\[
    \sum_{i=1}^{k} \left(2 l_i\epsilon\cdot n_0\right)
    =\epsilon\cdot  \BIGO{\sum_{i=1}^{k} l_i}
    =\epsilon\cdot  \BIGO{OPT + \sum_{i=1}^{k} m_i}=\BIGO{\epsilon\cdot OPT},
\]
where the last step is due to \cref{corollary:sp-lower-bound} applied on each subsequence: in particular, the optimal path length visiting all the regions in subsequence $i$ has length at least $\Omega(m_i)$, so summing this inequality over all subsequences, we have $\sum_{i=1}^{k} m_i \leq \BIGO{OPT}$.

We have now reduced the problem to the case where each region has only finitely many points.
We finish with dynamic programming. Since we have $\BIGO{\frac 1{\epsilon^{d-1}}}$ points on each of the $n$ regions, the runtime is
$\BIGO{\frac n{\epsilon^{2d-2}}}$, as desired.
\end{proof}

\begin{theorem}\label{thm:fat-non-isect-2d-1}
There is an $\BIGO{\frac{n}{\eps}}$-time algorithm for touring disjoint convex fat bodies in $2$ dimensions.
\end{theorem}

\begin{proof}
This is almost the same as \cref{thm:fat-non-isect-d}, where $\BIGO{\frac{1}{\epsilon^{d-1}}} = \BIGO{\frac{1}{\epsilon}}$ points are placed on each body,
except that we use \cref{corollary:linear-touring} to more efficiently solve the case where each region is a finite point set.
\end{proof}


\begin{theorem}\label{thm:fat-non-isect-2d-2}
There is an $\BIGO{\frac{n^2}{\sqrt{\eps}} }$-time algorithm for touring disjoint convex fat bodies in $2$ dimensions.
\end{theorem}

\begin{proof}
\cref{thm:fat-non-isect-2d-1} through the construction of \cref{thm:fat-non-isect-d} places $\BIGO{\frac{1}{\epsilon}}$ points on an arc of length $R$ on each convex fat body to guarantee additive error $\leq \epsilon R$. We can achieve the same additive error using a convex polygon with $\BIGO{\epsilon^{-1/2}}$ vertices using \cref{lemma:half-eps-intersect}. Then, recall that \cite{starburst-improve} gives an $\BIGO{|V|n}$-time exact algorithm for touring convex polygons, so we can recover a solution in $\BIGO{|V|n} = \BIGO{(n \cdot \epsilon^{-1/2}) \cdot n}$ time.
\end{proof}

\section{Balls}
\label{section:balls}

We can improve the results in previous sections by discretizing the surfaces non-uniformly, placing fewer points on areas of each hypersphere that are farther away from the previous and next ball in the sequence. This reduces the dependence on $\eps$ by a square root compared to \cref{thm:fat-non-isect-d} and \cref{thm:fat-non-isect-2d-1}. Omitted proofs for this section may be found in \cref{sec:omitted-proofs-balls}. We first state the results:

\nonIsectDisksTwoD

\begin{theorem}\label{thm:non-isect-balls}
There is an $\BIGO{\frac{n}{\eps^{d-1}}\log^2\frac{1}{\eps}+\frac{1}{\eps^{2d-2}}}$-time algorithm for touring disjoint balls in $d$ dimensions.

\end{theorem}

The crucial lemma we use for these results follows. We defer its proof to the appendix.

\begin{lemma}\label{lemma:local-global-opt}
A tour of disjoint balls is globally optimal if and only if for each intermediate ball, the tour either passes straight through the ball or perfectly reflects off its border (see \Cref{fig:local-global-opt} for an example).
\end{lemma}

\begin{figure}[h]
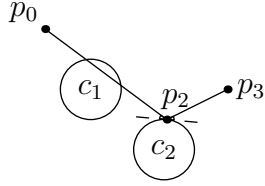

    \centering
    \begin{asy}
    import olympiad;
    
    unitsize(0.4cm);
    pair A = (0,0);
    pair C = (4,-3);
    pair D = (6,-2);
    draw(A--C--D);
    pair perp = unit((unit(A-C)+unit(D-C))/2);
    pair c2 = C-perp;
    draw(circle(c2,1));
    dot("$p_0$",A,NW);
    dot("$p_2$",C,1.2*dir(70));
    dot("$p_3$",D);
    pair c1 = (1.5,-2);
    label("$c_1$",c1);
    label("$c_2$",c2);
    draw(circle(c1,1));
    draw((C+dir(90)*perp)--(C-dir(90)*perp),dashed);
    draw(anglemark(A,C,C+dir(90)*perp));
    draw(anglemark(C-dir(90)*perp,C,D));
    \end{asy}
    
    \caption{\Cref{lemma:local-global-opt}: An optimal tour of two unit disks. The tour starts at $p_0$, passes through $c_1$, reflects off $c_2$ at $p_2$, and ends at $p_3$. The optimality of this tour may be certified by setting the dual variables $z_0=z_1=\unit(p_2-p_0)$ and $z_2=\unit(p_3-p_2)$ as defined in the proof of \cref{lemma:socp} in the appendix.}
    \label{fig:local-global-opt}
\end{figure}

We start with the special case of unit disks and then generalize to non-unit disks (\cref{thm:non-isect-disks}). First, we provide intuition through a simple example where $n=1$ and $R_1$ is a line.

\begin{example}\label{example:simple-reflect}
Given start and endpoints $p_0=(-1,1)$ and $p_2=(1,1)$, select $p_1$ from the $x$-axis such that $OPT=\vmag{p_0-p_1}+\vmag{p_1-p_2}$ is minimized.
\end{example}

\begin{proof}[Solution]
To solve this exactly, choose $p_1=(0,0)$ such that the path perfectly reflects off the $x$-axis. This gives $OPT=2\sqrt 2$. 

Now suppose that we are only interested in an approximate solution. Tile the $x$-axis with points at regular intervals such that every two consecutive points are separated by distance $d$, and round $p_1$ to the closest such point $p_1'$. Since $\vmag{p_1-p_1'}\le d$,
\begin{align*}
OPT'&\triangleq \vmag{p_0-p_1'}+\vmag{p_1'-p_2}\\
&\le \sqrt{1+(1-d)^2}+\sqrt{1+(1+d)^2}\le \sqrt{2-2d+d^2}+\sqrt{2+2d+d^2}\\
&\le \sqrt 2(1-d/2+1+d/2+\BIGO{d^2})\le 2\sqrt 2(1+\BIGO{d^2}).
\end{align*}

So, to attain $OPT'\le (1+\eps)OPT$, it suffices to take $d=\Theta(\sqrt \eps)$ rather than $d=\Theta(\eps)$ because $p_1'-p_1$ is parallel to the $x$-axis. We can apply a similar idea to replace the middle region with a point set when $R_1$ is a circle rather than a line since circles are locally linear. However, this doesn't quite work when either $\vmag{p_0-p_1}$ or $\vmag{p_1-p_2}$ is small. For example, if $p_0$ was very close to the $x$-axis (say, $p_0=(-d,d)$) then rounding $p_1$ to the nearest $p_1'$ could cause $OPT'$ to increase by $\Theta(d)\gg d^2$. So when we replace each circle with a point set, we need to be careful about how we handle two circles that are close to touching; the solution is to space points more densely near where they touch.
\end{proof}

\begin{theorem}\label{thm:non-isect-unit-disks}
There is an $\BIGO{\frac{n}{\sqrt{\eps}}\log\frac{1}{\eps}}$-time algorithm for touring disjoint unit disks.
\end{theorem}

\begin{proof}
We describe how to place a set of $\BIGO{\frac{1}{\sqrt \eps}\log \frac{1}{\eps}}$ points $S_i$ on each unit circle $c_i$ so that the length of an optimal path increases by at most $\BIGO{n\eps}$ after rounding each $p_i$ to the nearest $p_i'\in S_i$.

Define $\unit(x)=\frac{x}{\vmag{x}}$. Let $o_i\triangleq p'_{i}-p_i$ for all $i\in [0,n+1]$ (note that $o_0=o_{n+1}=0$), where $o$ stands for offset. Also, define vectors 
\begin{equation*}
d_i\triangleq p'_{i+1}-p'_i=p_{i+1}+o_{i+1}-p_i-o_i
\end{equation*}
and scalars
\begin{equation*}
a_i\triangleq d_i\cdot \unit(p_{i+1}-p_i)=\vmag{p_{i+1}-p_i}+(o_{i+1}-o_i)\cdot \unit(p_{i+1}-p_i), 
\end{equation*}
where $a_i$ is the component of $d_i$ along the direction of $p_{i+1}-p_i$. Then the total path length after rounding each $p_i$ to $p_i'$ is:
\begin{align*}
\sum_{i=0}^n\vmag{d_i}
&=\sum_{i=0}^n\sqrt{[d_i\cdot \unit(p_{i+1}-p_i)]^2+[d_i\cdot \unit(p_{i+1}-p_i)^{\perp}]^2} \nonumber \\
&=\sum_{i=0}^n\sqrt{a_i^2+[(o_{i+1}-o_i)\cdot \unit(p_{i+1}-p_i)^{\perp}]^2} \nonumber\\
&=\sum_{i=0}^n\left[a_i+\left(\sqrt{a_i^2+[(o_{i+1}-o_i)\cdot \unit(p_{i+1}-p_i)^{\perp}]^2}-a_i\right)\right] \nonumber\\
&=OPT+\sum_{i=1}^n \overbrace{o_i\cdot (\unit(p_i-p_{i-1})-\unit(p_{i+1}-p_i))}^{\text{extra}_1(i)} \nonumber\\
&\quad +\sum_{i=0}^n\overbrace{\left(\sqrt{a_i^2+[(o_{i+1}-o_i)\cdot \unit(p_{i+1}-p_i)^{\perp}]^2}-a_i\right)}^{\text{extra}_2(i)}\\
&=OPT+\sum_{i=1}^n\text{extra}_1(i)+\sum_{i=0}^n\text{extra}_2(i).
\end{align*}
We defer the construction of the sets $S_i$ so that both extra terms are small to \cref{lemma:smart-choose}. Then we can finish with dynamic programming (\cref{corollary:linear-touring}).
\end{proof}

\begin{lemma}\label{lemma:smart-choose}
It is possible to choose $S_i$ in the proof of \Cref{thm:non-isect-unit-disks} such that $|S_i|\le \BIGO{\frac{1}{\sqrt \eps}\log \frac{1}{\eps}}$, $\text{extra}_1(i)\le \BIGO{\eps}$, and $\text{extra}_2(i)\le \BIGO{\eps}$ for all $i$.
\end{lemma}

\begin{proof}
First, we present the construction. For every pair of adjacent disks $i$ and $i+1$ we describe a procedure to generate points on their borders. Then we set $S_i$ to be the union of the generated points on the border of disk $i$ when running the procedure on disks $(i,i+1)$, and the generated points on the border of disk $i$ when running the procedure on disks $(i-1,i)$. Finally, we show that $\text{extra}_1(i)$ and $\text{extra}_2(i)$ are sufficiently small for all $i$ for our choice of $S_i$.

\subparagraph*{Procedure} Reorient the plane that $c_i=(0,y)$ and $c_{i+1}=(0,-y)$ for some $y>1$. 
 Let $\spacing\colon \mathbb{R}_{\ge 0}\to \mathbb{R}_{>0}$ be a function that is nonincreasing with respect to $|\phi|$ that we will define later. Given $\spacing$, we use the following process to add points to $S_i$ (and symmetrically for $S_{i+1}$):

\begin{enumerate}
    \item Set $\phi=0$.
    \item While $\phi\le \pi$:
    \begin{itemize}
        \item Add $(\sin\phi, y-\cos\phi)$ to $S_i$.
        \item $\phi \mathrel{{+}{=}} \spacing(\phi)$.
    \end{itemize}
    \item Repeat steps 1-2 but for $\phi$ from $0$ to $-\pi$.
\end{enumerate}
This procedure has the property that for any $\phi \in [-\pi,\pi]$, the point $(\sin\phi, y-\cos\phi)$ is within distance $\spacing(|\phi|)$ of some point in $S_i$. In particular, if the optimal path has $p_i=(\sin\phi_i,y-\cos\phi_i)$ then it is guaranteed that $\vmag{o_i}\le \spacing(\phi_i)$. To compute $|S_i|$, note that as long as $\spacing(\phi)$ is sufficiently smooth that $\frac{\spacing(\phi)}{\spacing\paren{\phi+\spacing(\phi)}}=\Theta(1)$ for all $\phi$, the number of points added to $S_i$ will be at most a constant factor larger than the value of the definite integral $\int_{-\pi}^{\pi}\frac{1}{\spacing(\phi)}\, d\phi$. 

Next, we construct $\spacing$ so that $|S_i|=\BIGO{\frac{1}{\sqrt \eps}\log \frac{1}{\eps}}$. Intuitively, by \cref{example:simple-reflect}, we should have $\spacing(\phi)=\Theta(\eps)$ closer to circle $i+1$ (when $\phi\approx 0$) and $\spacing(\phi)=\Theta(\sqrt \eps)$ farther from circle $i+1$ (when $\phi=\Theta(1)$). Thus, we set $\spacing(\phi)=\max(\eps,\sqrt{\eps}\phi)$. The total number of added points is on the order of:
\begin{align*}
\int_0^{\pi}\frac{1}{\spacing(\phi)}\, d\phi&=\frac{1}{\sqrt \eps}\left(\int_0^{\sqrt \eps}\frac{1}{\sqrt \eps}\, d\phi+\int_{\sqrt \eps}^{\pi}\frac{1}{\phi}\, d\phi\right)\\
&=\frac{1}{\sqrt \eps}\left(1+\log\left(\frac{\pi}{\sqrt \eps}\right)\right)\le \BIGO{\frac{1}{\sqrt \eps}\log \frac{1}{\eps}}.
\end{align*}

Finally, we show that both extra terms are small for our choice of $S_i$.

\proofsubparagraph*{Part 1: $\text{extra}_1(i)$.} 

We note that $\unit(p_i-p_{i-1})-\unit(p_{i+1}-p_i)$ must be parallel to $p_i-c_i$ for an optimal solution $p$. To verify this, it suffices to check the two possible cases from \Cref{lemma:local-global-opt}:
\begin{enumerate}
    \item The points $p_{i-1},p_i,p_{i+1}$ are collinear, in which case $\unit(p_i-p_{i-1})-\unit(p_{i+1}-p_i)=0$.
    \item The path reflects perfectly off circle $i$, in which case $\unit(p_i-p_{i-1})-\unit(p_{i+1}-p_i)$ is parallel to $p_i-c_i$.
\end{enumerate}

If we ensure that $\text{spacing}(\phi)\le \sqrt \eps$ for all $\phi$, then $|o_i\cdot \unit(p_i-c_i)|\le \eps$ because $o_i$ is always nearly tangent to the circle centered at $c_i$ at point $p_i$. The conclusion follows because $\text{extra}_1(i)\le 2|o_i\cdot \unit(p_i-c_i)|\le 2\eps$.

\proofsubparagraph*{Part 2: $\text{extra}_2(i)$.}

We upper bound $\text{extra}_2(i)$ by the sum of two summands, the first associated only with $o_i$ and the second associated only with $o_{i+1}$.

\begin{claim}\label{claim:splitting-extra}
Letting $\text{ycoord}(\cdot)$ denote the $y$-coordinate of a point,
\begin{equation*}
  \text{extra}_2(i)\le 2\cdot \left( \min\left(\vmag{o_i},\frac{4\vmag{o_i}^2}{\text{ycoord}(p_i)}\right)+\min\left(\vmag{o_{i+1}},\frac{4\vmag{o_{i+1}}^2}{-\text{ycoord}(p_{i+1})}\right)\right). 
\end{equation*}
\end{claim}

\begin{claimproof}
We do casework based on which term is smaller on each of the $\min$s.
\begin{enumerate}
    \item $\vmag{o_i}\ge \frac{\text{ycoord}(p_i)}{4}$, $\vmag{o_{i+1}}\ge \frac{-\text{ycoord}(p_{i+1})}{4}$
    
    The result, $\text{extra}_2(i)\le 2(\vmag{o_i}+\vmag{o_{i+1}})$, follows by summing the following two inequalities:
    
    \begin{align*}
        &\sqrt{a_i^2+[(o_{i+1}-o_i)\cdot \unit(p_{i+1}-p_i)^{\perp}]^2}-\vmag{p_{i+1}-p_{i}}\\
        &=\vmag{p_{i+1}-p_i+o_{i+1}-o_i}-\vmag{p_{i+1}-p_{i}} \\
        &\le \vmag{o_i}+\vmag{o_{i+1}}
    \end{align*}
    and $\vmag{p_{i+1}-p_{i}}-a_i\le \vmag{o_i}+\vmag{o_{i+1}}$.
    
    \item $\vmag{o_i}\le \frac{\text{ycoord}(p_i)}{4}$, $\vmag{o_{i+1}}\le \frac{-\text{ycoord}(p_{i+1})}{4}$
    
    Then $\vmag{o_i}, \vmag{o_{i+1}}\le \frac{\vmag{p_{i+1}-p_{i}}}{4}$ so $a_i\ge \frac{\vmag{p_{i+1}-p_i}}{2}$, and
    \begin{align*}
        \text{extra}_2(i)&\le \frac{\vmag{o_{i+1}-o_i}^2}{2a_i} \le \frac{2(\vmag{o_{i+1}}^2+\vmag{o_i}^2)}{2a_i}\\
        &\le 2\cdot \frac{\vmag{o_{i+1}}^2+\vmag{o_i}^2}{\vmag{p_i-p_{i+1}}}\le 2\cdot \left(\frac{\vmag{o_i}^2}{\text{ycoord}(p_i)}+\frac{\vmag{o_{i+1}}^2}{-\text{ycoord}(p_{i+1})}\right).
    \end{align*}
    
    \item $\vmag{o_i}\le \frac{\text{ycoord}(p_i)}{4}$, $\vmag{o_{i+1}}\ge \frac{-\text{ycoord}(p_{i+1})}{4}$
    
    Define $\text{extra}'(i)$ to be the same as $\text{extra}_2(i)$ with $o_{i+1}$ set to 0. Then
    \begin{align*}
        \text{extra}'(i)&\triangleq \vmag{p_{i+1}-p_i-o_i}-(\vmag{p_{i+1}-p_i}-o_i\cdot \unit(p_{i+1}-p_i))\\
        &=\sqrt{(\vmag{p_{i+1}-p_i}-o_i\cdot \unit(p_{i+1}-p_i))^2+[o_i\cdot \unit(p_{i+1}-p_i)^{\perp}]^2}\\
        &\quad -(\vmag{p_{i+1}-p_i}-o_i\cdot \unit(p_{i+1}-p_i))\\
        &\le \frac{\vmag{o_i}^2}{2\cdot \frac{3}{4}\vmag{p_i-p_{i+1}}}\le \frac{\vmag{o_i}^2}{2\cdot \frac{3}{4}\cdot \text{ycoord}(p_i)}
    \end{align*}
    and by similar reasoning as case 1, $\text{extra}_2(i)-\text{extra}'(i)\le 2\vmag{o_{i+1}}$.
    
    \item $\vmag{o_i}\ge \frac{\text{ycoord}(p_i)}{4}$, $\vmag{o_{i+1}}\le \frac{-\text{ycoord}(p_{i+1})}{4}$
    
    Similar to case 3. \qedhere
\end{enumerate}
\end{claimproof}

Now that we have a claim showing an upper bound on $\text{extra}_2(i)$, it remains to show that $\min\left(\vmag{o_i},\frac{\vmag{o_i}^2}{\text{ycoord}(p_i)}\right)\le \BIGO{\eps}$ for our choice of $\spacing$. Indeed, when $\phi\le \sqrt\eps$ we have $\vmag{o_i}\le \spacing(\phi)\le \eps$, while for $\phi>\sqrt \eps$ we have $\frac{\vmag{o_i}^2}{\text{ycoord}(p_i)}\le \BIGO{\frac{\spacing(\phi)}{\phi^2}}\le \BIGO{\eps}$.
\end{proof}

With small modifications to the proof of \Cref{lemma:smart-choose}, we have the following corollary:

\begin{corollary}\label{corollary:non-unit}
Consider the case of non-unit disks. If the $i$th disk has radius $r_i$, then we can place $\BIGO{\frac{1}{\sqrt \eps_i}\log \frac{1}{\eps_i}}$ points on its border such that the additive error associated with $c_i$ \textemdash specifically, $\text{extra}_1(i)$ plus the components of $\text{extra}_2(i-1)$ and $\text{extra}_2(i)$ associated with $\vmag{o_i}$ \textemdash is $\BIGO{r_i\eps_i}$. Consequently,
$OPT+\sum_{i=1}^n\text{extra}_1(i)+\sum_{i=0}^n\text{extra}_2(i)\le OPT+\sum_{i=1}^nr_i\eps_i.$
\end{corollary}

Now, we finally prove \cref{thm:non-isect-disks,thm:non-isect-balls}.

\begin{proof}[Proof of \cref{thm:non-isect-disks} (Non-Unit Disks)]
We first present a slightly weaker result, and then show how to improve it. Recall that by \cref{corollary:sp-lower-bound}, the $i$th largest disk has radius $\BIGO{\frac{OPT}{i}}$ for $i\ge 3$. So if we set $\eps_i=\eps'=\frac{\eps}{\log n}$ for each of the $i$th largest disks for $i\ge 3$, the total additive error contributed by these disks becomes
\begin{equation*}
\BIGO{\sum_{i=3}^n\frac{OPT}{i}\cdot \eps_i}\le \BIGO{OPT\cdot \eps'\cdot \sum_{i=3}^n\frac{1}{i}}\le \BIGO{\eps OPT}
\end{equation*}
by \cref{corollary:non-unit}. For the two largest disks, we use the previous naive discretization (placing $\BIGO{\frac{1}{\eps}}$ points uniformly on the intersection of the circles with a square of side length $\BIGO{OPT}$ centered about the starting point). We may assume we have already computed a constant approximation to $OPT$ in $\BIGO{n}$ time by applying \cref{thm:fat-non-isect-2d-1} with $\eps=1$. After selecting the point sets, we can finish with \cref{corollary:linear-touring}. The overall time complexity is
$
\BIGO{\frac{n}{\sqrt{\eps'}}\log \frac{1}{\eps'}+\frac{1}{\eps}}\le \BIGO{\frac{n\sqrt {\log n}}{\sqrt \eps}\log \paren{\frac{\log n}{\eps}}+\frac{1}{\eps}}$.

We can remove the factors of $\log n$ by selecting the $\eps_i$ to be an increasing sequence. Set $\eps_i=\Theta\paren{\frac{\eps i^{2/3}}{n^{2/3}}}$ for each $i\in [3,n]$ such that more points are placed on larger disks. Then the total added error remains
\begin{align*}
\BIGO{OPT\cdot \paren{\eps+\sum_{i=3}^n\frac{\eps_i}{i}}}&=
\BIGO{OPT\cdot \paren{\eps+\sum_{i=3}^n\frac{1}{i}\cdot \frac{\eps i^{2/3}}{n^{2/3}}}}\\
&=\BIGO{OPT\eps\cdot \paren{1+n^{-2/3}\cdot \sum_{i=3}^ni^{-1/3}}}\le \BIGO{OPT\eps},
\end{align*}
and the factors involving $\log n$ drop out from the time complexity:
\begin{align*}
\BIGO{\sum_{i=3}^n\frac{1}{\sqrt \eps_i}\log\paren{\frac{1}{\eps_i}}+\frac{1}{\eps}}&\le \BIGO{\int_{i=3}^n\frac{1}{\sqrt \eps}n^{1/3}i^{-1/3}\log\paren{ \frac{n^{2/3}}{i^{2/3}\eps}}di +\frac{1}{\eps}}\\
&\le \BIGO{\frac{3n^{1/3}}{2\sqrt{\epsilon}}i^{2/3}\paren{\log{\frac{n^{2/3}}{i^{2/3}\epsilon}}+1}\bigg\rvert_{3}^{n}+\frac{1}{\epsilon}} \\
&\le \BIGO{\frac{n}{\sqrt \eps}\log\paren{\frac{1}{\eps}}+\frac{1}{\eps}}. \qedhere
\end{align*}
\end{proof}

We note that under certain additional assumptions, the time complexity of \Cref{thm:non-isect-disks} can be slightly improved. We summarize these in the following corollary, which we state without proof.

\begin{corollary}
If any of the following conditions hold:
\begin{itemize}
    \item the two largest disks are not adjacent in the order
    \item the two largest disks are separated by distance $\Omega(OPT)$
    \item the second-largest disk has radius $\BIGO{OPT}$
\end{itemize}
then a generalization of \Cref{lemma:smart-choose} may be applied to remove the $\frac{1}{\eps}$ term from the time complexity of \Cref{thm:non-isect-disks}.
\end{corollary}

To extend to multiple dimensions, we generalize the construction from \Cref{lemma:smart-choose}.

\begin{proof}[Proof of \cref{thm:non-isect-balls} (Balls)]
 As in \Cref{lemma:smart-choose}, set $\spacing(\phi)=\max(\eps,\sqrt \eps \phi)$ for a point $p_i$ satisfying $m\angle p_ic_ic_{i+1}=\phi$, meaning that there must exist $p_i'\in S_i$ satisfying $\vmag{p_i-p_i'}\le r_i\cdot \spacing(\phi)$. The total number of points $|S_i|$ placed on the surface of a $d$-dimensional sphere is proportional to
\begin{align*}
\int_0^{\pi}\frac{\sin^{d-2}(\phi)}{\spacing(\phi)^{d-1}}\, d\phi
&\le \frac{1}{(\sqrt \eps)^{d-1}}\int_0^{\pi} \frac{\phi^{d-2}}{\max(\sqrt \eps,\phi)^{d-1}}\, d\phi\\
&=\frac{1}{\eps^{(d-1)/2}}\paren{\int_0^{\sqrt \eps}\frac{\phi^{d-2}}{(\sqrt \eps)^{d-1}}\, d\phi + \int_0^{\sqrt \eps}\frac{1}{\phi}\, d\phi} \\
&\le \BIGO{\frac{1}{\eps^{(d-1)/2}}\log \frac{1}{\eps}}.
\end{align*}
where the derivation of the integration factor $\sin^{d-2}(\phi)$ can be found in \cite{unit-ball}. 

It remains to describe how to space points so that they satisfy the given spacing function. For each spacing $s=\eps, 2\eps, 4\eps, \ldots, \sqrt{\eps}$, we can find a $d$-dimensional hypercube of side length $O(s/\sqrt \eps)$ that encloses all points on the hypersphere with required spacing at most $2s$. Evenly space points with spacing $s$ across the surface of this hypercube according to \Cref{lemma:equal-spacing-points}, and project each of these points onto the hypersphere. There are a total of $\BIGO{\log \frac{1}{\eps}}$ values of $s$, and each $s$ results in $\BIGO{\frac{1}{\eps^{(d-1)/2}}}$ points being projected onto the hypersphere, for a total of $\BIGO{\frac{1}{\eps^{(d-1)/2}}\log \frac{1}{\eps}}$ points.
\end{proof}



\bibliography{bib.bib}

\appendix

\section{Appendix}

Organization: The three subsections contain omitted proofs from sections 2, 3, and 4, respectively.

\subsection{Convex bodies: omitted proofs}\label{sec:omitted-proofs-convex-bodies}

\begin{proof}[Proof of \cref{lemma:project-contract}]
Define $c_1\triangleq \text{closest}_C(p_1)$ and $c_2\triangleq \text{closest}_C(p_2)$. Since $C$ is convex, $C$ must contain all points on the segment connecting $c_1$ and $c_2$. Thus, it must be the case that $(p_1-c_1)\cdot (c_2-c_1)\le 0$, or some point on the segment connecting $c_1$ and $c_2$ would be closer to $p_1$ than $c_1$. Similarly, it must be the case that $(p_2-c_2)\cdot (c_2-c_1)\ge 0$. To finish,
\begin{align*}
(p_2-p_1)\cdot (c_2-c_1)&\ge (p_2-c_2+c_2-c_1+c_1-p_1)\cdot (c_2-c_1)\\
&\ge (c_2-c_1)\cdot (c_2-c_1)\ge \vmag{c_2-c_1}^2,
\end{align*}
implying $\vmag{p_2-p_1}\ge \vmag{c_2-c_1}$.
\end{proof}

\begin{proof}[Proof of \cref{lemma:equal-spacing-points}]
Consider some point $p \in (\partial C) \cap \mathcal H$. By the convexity of $C$ and because $p$ lies on its boundary, there exists some unit vector $\vec{v}$ such that $p \cdot \vec{v} \geq x \cdot \vec{v}$ for any $x \in C$. Let the intersection of $\mathcal H$ with the ray starting at $p$ and going in the direction of $\vec{v}$  be the point $s$. Notice that $\closest_C(s) = p$. 

Now, define $s' \in S_{\mathcal H}$ to be the closest such point to $s$. From our construction of $S_{\mathcal H}$, $\vmag{s-s'} \leq r\eps$. From our construction of $S$, $\closest_C(s') \in S$, and since $\closest_C(s) = p$, after applying \cref{lemma:project-contract}, we have $\vmag{p-\closest_C(s')} \leq r\eps$.
\end{proof}

\begin{proof}[Proof of \cref{lemma:pseudo-approx}]
First, we show that there exists a path $p'$ touring the $S_i$ that is nearly as short as the shortest path touring the $R_i$.

\begin{claim*}
If $OPT\le L_{APPROX}$, there exists an increasing sequence $0 = z_0 < z_1 < \dots < z_l = n+1$ and points $p_{z_i}'$ such that $p_{z_i}' \in S_{z_i}$ for all $i$ and $p_{z_i}' \in R_{j}$ for all $i, j$ satisfying $0 \leq i \leq l-1, z_i \leq j < z_{i+1}$, and $\sum_{i=0}^{l-1} \vmag{p_{z_i}'-p_{z_{i+1}}'} \leq OPT+\gamma L_{APPROX}$. Additionally, if the regions are disjoint, there exists a sequence that satisfies the above conditions that also satisfies $z_i = i$ for all $i$.
\end{claim*}

\begin{claimproof}
Define $D(z_t, p_{z_t}')$ to be the minimum distance for touring regions $R_{z_t}, R_{z_t+1}, \dots, R_n$ starting at $p_{z_t}'$ and ending at $p_{n+1}$. We show by induction that for all $t\le l$, there exists a sequence $p_{z_0},\dots,p_{z_t}$ satisfying the invariant
\begin{equation}
    \sum_{i=0}^{t-1} \vmag{p_{z_i}'-p_{z_{i+1}}'} + D(z_t, p_{z_t}) \leq OPT+\frac{t}{2n} \cdot \gamma L_{APPROX}.\label{ineq:induct}
\end{equation}
Note that for $t=l$, (\ref{ineq:induct}) implies
\begin{equation*}
    \sum_{i=0}^{l-1} \vmag{p_{z_i}'-p_{z_{i+1}}'} \leq OPT + \frac{l}{2n} \cdot \gamma L_{APPROX} \leq OPT + \gamma L_{APPROX},
\end{equation*}
as desired.

\noindent
\textbf{Base Case:} First, define $z_0 = 0, p_{z_0}' = p_0$. The invariant holds for $t=0$. 

\noindent
\textbf{Inductive Step:} Suppose we have already constructed $z_0, z_1, \dots, z_t$ and $p_{z_0}', p_{z_1}', \dots, p_{z_t}'$. It remains to show how to construct $z_{t+1}, p_{z_{t+1}}'$ from $z_{t}, p_{z_t}'$ such that the invariant is maintained.

Let $z_{t+1} \leq n$ be the smallest integer greater than $z_t$ such that $p_{z_t}' \notin R_{z_{t+1}}$. If all $R_i$ are disjoint, then clearly $z_{t+1} = z_t+1$. If no such $z_{t+1}$ exists, then we set $z_{t+1} = n+1$. Consider an optimal path $q_{z_t}, q_{z_t+1}, q_{z_t+2}, \dots, q_{n+1}$ for touring regions $R_{z_t}, R_{z_t+1}, \dots, R_n$ starting at $q_{z_t} = p_{z_t}'$ and ending at $q_{n+1} = p_{n+1}$, with total length $D(z_t, p_{z_t}')$. Because $q_{z_t} = p_{z_t}' \in R_{j}$ for all $z_t \leq j < z_{t+1}$, we can consider an optimal path that satisfies $q_{z_t} = q_{z_t+1} = q_{z_t+2} = \dots = q_{z_{t+1}-1}$, and $q_{z_{t+1}}$ lies on the boundary of $R_{z_{t+1}}$. In other words, the optimal path does not need to move from its starting point if it is contained within some prefix of the regions, and it can always choose $q_{z_{t+1}}$ to be its first point of contact with $R_{z_{t+1}}$. Therefore $q_{z_{t+1}} \in \partial C_{z_{t+1}, w}$ for some $1\le w\le |R_{z_{t+1}}|$. 

Furthermore, a path exists from $p_0$ to $p_{z_t}' = q_{z_t}$ to $p_{n+1}$ with length 
\begin{equation*}
    \sum_{i=0}^{t-1} \vmag{p_{z_i}'-p_{z_{i+1}}'} + D(z_t, p_{z_t}') \leq OPT+\frac{t}{2n}\cdot \gamma L_{APPROX} \leq OPT + L_{APPROX} \leq 2L_{APPROX},
\end{equation*}
indicating that $\vmag{q_{z_{t+1}}-p_0}\le 2L_{APPROX}$, which in turn implies $q_{z_{t+1}}\in \mathcal H$.

Thus, $q_{z_{t+1}}\in (\partial C_{z_{t+1}, w})\cap \mathcal H$, so by \cref{lemma:equal-spacing-points} there exists some point $p_{z_{t+1}}' \in S_{z_{t+1}}$ such that $\vmag{p_{z_{t+1}}'-q_{z_{t+1}}} \leq \epsilon \cdot r = \frac{\gamma}{16n} \cdot 4 L_{APPROX}$. Now, we can show that the invariant holds for $t+1$:
\begin{align}
    &\sum_{i=0}^{t} \vmag{p_{z_i}'-p_{z_{i+1}}'} + D(z_{t+1}, p_{z_{t+1}}') \nonumber \\
    &\leq 
    \sum_{i=0}^{t} \vmag{p_{z_i}'-p_{z_{i+1}}'} + \vmag{p_{z_{t+1}}'-q_{z_{t+1}}}+\sum_{i=z_{t+1}}^{n} \vmag{q_i-q_{i+1}}  \label{align:long-series-1} \\
    &\le \sum_{i=0}^{t-1} \vmag{p_{z_i}'-p_{z_{i+1}}'} +
    \vmag{p'_{z_t}-q_{z_{t+1}}} + 2\vmag{p_{z_{t+1}}'-q_{z_{t+1}}}+\sum_{i=z_{t+1}}^{n} \vmag{q_i-q_{i+1}} \label{align:long-series-2} \\
    &\leq \sum_{i=0}^{t-1} \vmag{p_{z_i}'-p_{z_{i+1}}'} + \vmag{p'_{z_t}-q_{z_{t+1}}} + 2\eps r +\sum_{i=z_{t+1}}^{n} \vmag{q_i-q_{i+1}} \nonumber\\
    &= 2\eps r + \sum_{i=0}^{t-1} \vmag{p_{z_i}'-p_{z_{i+1}}'}  + D(z_t, p_{z_t}') \nonumber \\
    &\leq 2 \eps r + OPT + \frac{t}{2n} \cdot \gamma L_{APPROX}\nonumber \\
    &= OPT + \frac{\gamma L_{APPROX}}{2n} + \frac{t}{2n} \cdot \gamma L_{APPROX} = OPT + \frac{t+1}{2n} \cdot \gamma L_{APPROX}. \nonumber
\end{align}
(\ref{align:long-series-1}) follows from $p_{z_{t+1}}', q_{z_{t+1}+1}, \dots, q_{n+1}$ being a valid tour of regions $R_{z_{t+1}}, R_{z_{t+1}+1}, \dots, R_n$, while (\ref{align:long-series-2}) follows from the triangle inequality.

Recall that if all $R_i$ are disjoint, then $z_{t+1} = z_t+1 \implies z_i = i$ for all $i$.
\end{claimproof}

It remains to show that we can recover a path touring the $S_i$ that is at least as good as $p'$. For each point $p \in S_i$ for some $i$, define $successor(p)$ to be the minimum $j > i$ such that $p \notin R_j$. Notice that $successor(p_{z_i}') = z_{i+1}$. Now, define $pred(S_j)$ for each $S_j$ to be the set of all points $p$ such that $successor(p) = j$.

We use dynamic programming: first, set $dp(p_0) = 0$. Then, iterate over the sets $S_i$ in increasing $i$. For each point $x \in S_i$, we will set $dp(x)\triangleq \min_{y \in pred(S_j)}dp(y)+\vmag{x-y}$. By storing the optimal transition for each point, we can recover a path $q_{0}, q_1, \dots, q_{l_q}$ that tours the $R_i$ regions. Because $successor(p_{z_i}') = z_{i+1}$ for the path $(p_i')$, we must have that the path recovered from dynamic programming has length at most $\sum_{i=0}^{l-1} \vmag{p_{z_i}'-p_{z_{i+1}}'}$.

If the $R_i$ are disjoint, then $pred(S_j) = S_{j-1}$ and the total number of transitions in the DP is given by
\begin{equation*}
   \BIGO{\sum_{i=0}^{n} |S_i||S_{i+1}|} = \BIGO{\sum_{i=0}^{n} |R_i|\paren{\frac{n}{\gamma}}^{d-1}\cdot |R_{i+1}| \paren{\frac{n}{\gamma}}^{d-1}}, 
\end{equation*}
as desired. If the $R_i$ are possibly non-disjoint, notice that the total number of transitions in the DP is given by 
\begin{align*}
\sum_{i=0}^{n} \sum_{p \in S_i} \BIGO{|S_{successor(p)}|\paren{\frac{n}{\gamma}}^{d-1}} &\leq  \BIGO{\sum_{i=0}^{n} \sum_{p \in S_i} \max_{j} |R_j| \paren{\frac{n}{\gamma}}^{d-1}} \\
&\leq \BIGO{\sum_{i=0}^{n} |R_i|\paren{\frac{n}{\gamma}}^{d-1} \max_{j} |R_j| \paren{\frac{n}{\gamma}}^{d-1}}\\
&\leq \BIGO{\paren{\frac{n}{\gamma}}^{2d-2}\sum_{i=0}^{n} |R_i|\max_j|R_{j}|}.\qedhere
\end{align*}
\end{proof}

\begin{proof}[Proof of \cref{lemma:trivial-n-approx}]
Using the oracle to project $p_0$ onto each of the regions $R_1, R_2, \dots, R_n$ to obtain points $p_1, p_2, \dots, p_n$. Now, let $OPT$ be the total length of the optimal solution for the problem. Clearly, we must have $\vmag{p_i-p_0} \leq OPT$. Thus,
\begin{align*}
\sum_{i=0}^{n}\vmag{p_{i+1}-p_i} &\leq \vmag{p_1-p_0} + \sum_{i=1}^{n} \paren{\vmag{p_{i+1}-p_0}+\vmag{p_i-p_0}} \\
&\leq OPT + 2n \cdot OPT = (2n+1)OPT.\qedhere
\end{align*}
\end{proof}

\begin{proof}[Proof of \cref{lemma:non-intersect-d-dim-constant-approx}]
First, apply \cref{lemma:trivial-n-approx} to get a $1+2n$ approximation. Label the path length of this approximation $L_{0}$. Define $B_{0} = 1+2n$, where we know $\frac{L_0}{B_0} \le OPT\le L_0$. Our goal is to show that given some $L_{t}, B_t$, we can generate $L_{t+1}, B_{t+1}$ so that $\frac{L_{t+1}}{B_{t+1}} \le OPT\le L_{t+1}$, and $B_{t+1} \leq 2\sqrt{B_t}$. 

We apply \cref{lemma:pseudo-approx} with $\gamma = 1$, $L_{APPROX} = \frac{L_t}{\sqrt{B_t}}$. Let $L_{APPROX}'$ be the length of the optimal tour that visits $S_0, S_1, \dots, S_{n+1}$. There are two possible cases.

\begin{enumerate}
    \item $L'_{APPROX} \le 2L_{APPROX}$. In this case, we know that 
    $\frac{L_t}{B_t}\le OPT\le L'_{APPROX}\le 2L_{APPROX}=\frac{2L_t}{\sqrt B_t}$, so we can set $L_{t+1}=\frac{2L_t}{\sqrt B_t}$ and $B_{t+1}=2\sqrt {B_t}$.
    
    \item $L'_{APPROX} > 2L_{APPROX}$. In this case we know that $OPT>L_{APPROX}$, because \cref{lemma:pseudo-approx} guarantees that if $OPT \leq L_{APPROX}$, then $L'_{APPROX} \leq OPT + \gamma L_{APPROX} \leq 2L_{APPROX}$. Thus,
    $\frac{L_t}{\sqrt B_t}=L_{APPROX}< OPT\le L_t$,
    so we can set $L_{t+1}=L_t$ and $B_{t+1}=\sqrt B_t$.
\end{enumerate}

Thus, we can generate the sequences $(L_i)$ and $(B_i)$ until we reach some $B_l \leq 4$ for some $l \leq \BIGO{\log \log{n}}$ in $\BIGO{\log \log{n} \cdot n^{2d-2}\sum_{i=1}^{n-1}|R_i||R_{i+1}|}$ time. This gives us some $L_l$ such that $OPT \leq L_l \leq 4 \cdot OPT$, a constant approximation of $OPT$.
\end{proof}

\begin{proof}[Note for \cref{thm:2d-unions}]
\begin{sloppypar}
The bulk of the time for the second bound is spent computing $successor(p)$ for every one of the $\BIGO{\frac{n^2}{\eps}}$ points in the discretization, which could take $\Theta(n)$ calls to the oracle, contributing the factor of $\BIGO{\frac{n^3}{\eps}}$. On the other hand, the actual dynamic programming updates contribute only $\BIGO{\frac{n^2}{\eps}\log\paren{\frac{n^2}{\eps}}}$.\qedhere
\end{sloppypar}
\end{proof}

\begin{proof}[Proof of \cref{thm:nonconvex-polygons-2d}]
Here we describe how to modify the method of \cite{Mozafari2012TouringPA} to achieve the desired time complexity. The approach we describe in this paper is a more general method (which can be modified similarly).

As in our approach, the method of \cite{Mozafari2012TouringPA} involves a pseudo-approximation: the idea is to intersect every one of the $|V|$ edges with a disk of radius $L_{APPROX}$ centered at $p_0$, discretize every one of the $|V|$ edges into $\frac{n}{\eps}$ evenly spaced points, and then apply dynamic programming as described in the previous subsection. Naively, these DP transitions run in $\BIGO{\paren{\frac{|V|n}{\eps}}}^2$ time, but \cref{lemma:additive-voronoi} speeds these transitions up to $\BIGO{\frac{|V|n}{\eps}\log \frac{|V|}{\eps}}$ time. The first term in the time complexity corresponds to the time required to obtain a constant approximation by setting $\eps=1$, as described in \cref{lemma:pseudo-approx}.

When the regions can possibly intersect, we additionally need a data structure that will compute $successor(p)$ for any $p\in \partial R_i$ in $\BIGO{\log |V|}$ time. The construction of such a data structure contributes the additional term to the time complexity.

Now, we describe how to generate a separate data structure for each edge $e\in \delta R_i$ such that each data structure can answer $successor(p)$ for any $p\in \partial e$ in $\BIGO{\log |V|}$ time. This construction runs in $\BIGO{|V|^2\alpha(|V|)}$ time.

First, for each edge  $e\in \delta R_i$ we need to compute which parts of it belong to each other region $R_j$ where $j>i$. To do so, we need to know
\begin{itemize}
    \item For each such region, whether the endpoints of $e$ are contained within that region.
    \item The intersection points of $e$ with all such regions in sorted order along the edge (at most $|V|$, assuming non-degeneracy).
\end{itemize}
These quantities can be computed in $\BIGO{|V|^2}$ time due to Balaban \cite{segment-isect}. 

The intersection points partition $e$ into ranges such that the first $j>i$ such that a point belongs to $R_j$ is the same for all points within the range (see \cref{fig:v2-ds}). If we have computed the first $j>i$ for every such range, we can answer queries in $\BIGO{\log v}$ time via binary search.

\begin{figure}[h]
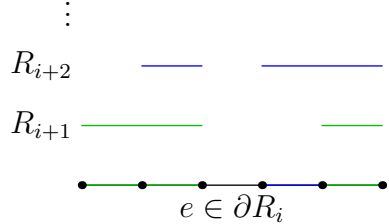

    \centering
    \begin{asy}
    size(5cm);
// dot((0,0));
// dot((1,0));
draw((0,0)--(1,0));
label("$e\in \partial R_i$", (0.5,0), S);

draw((0,0.2)--(0.4,0.2), heavygreen);
draw((0.8,0.2)--(1.0,0.2), heavygreen);
// dot((0,0.2), heavygreen);
// dot((0.4,0.2), heavygreen);
// dot((0.8,0.2), heavygreen);
// dot((1.0,0.2), heavygreen);
label("$R_{i+1}$", (0,0.2), W);

draw((0.2,0.4)--(0.4,0.4), blue);
draw((0.6,0.4)--(1.0,0.4), blue);
label("$R_{i+2}$", (0,0.4), W);
label("$\vdots$", (0,0.6), W);

draw((0,0)--(0.4,0), heavygreen);
draw((0.8,0)--(1.0,0), heavygreen);
draw((0.6,0)--(0.8,0), blue);
dot((0,0));
dot((0.2,0));
dot((0.4,0));
dot((0.6,0));
dot((0.8,0));
dot((1.0,0));
    \end{asy}
    \caption{Diagram for \cref{thm:nonconvex-polygons-2d}: Constructing a data structure for edge $e$. For each point $p\in e$, the goal is to identify the first $j>i$ such that $p$ belongs to $R_j$. Each point on $e$ has been colored with the color corresponding to the corresponding $R_j$. The intersection points of $e$ with $\delta R_{i+1}$ and $\delta R_{i+2}$ partition $e$ into five ranges, of which the first two and last one satisfy $j=i+1$, and the fourth satisfies $j=i+2$.}
    \label{fig:v2-ds}
\end{figure}

To do so, we start by iterating over all intervals of $e$ contained within $R_{i+1}$ and setting all ranges that they cover to have $successor=i+1$. This determines the answer for three of the ranges in \cref{fig:v2-ds}. Then do the same for $R_{i+2},R_{i+3},\ldots$ and so on. Note that there are $\BIGO{|V|}$ intervals in total.

Naively, such an implementation would run in $\BIGO{|V|^2}$ time; however, using Tarjan's disjoint set union data structure \cite{dsu}, we can speed up this process to $\BIGO{|V|\alpha(|V|)}$ time, where $\alpha$ is the inverse Ackermann function.  We assume an implementation of DSU that initially assumes that every range is its own representative, and supports both of the following operations in amortized $\alpha(|V|)$ time:
\begin{itemize}
    \item $\texttt{Find}(x)$: Return the representative of range $x$.
    \item $\texttt{Unite}(x,y)$: Given a range $x$ such that $\texttt{Find}(x)=x$, for each range that $x$ is a representative of, set its representative to be $\texttt{Find}(y)$.
\end{itemize}

For an interval covering ranges $[l,r]$ belonging to $R_j$, we use the following process to set the answers for every range it covers whose successor has not been set yet:
\begin{enumerate}
    \item Set $l=\texttt{Find}(l)$.
    \item If $l>r$, break.
    \item Set $successor(l)=j$
    \item Call $\texttt{Unite}(l,l+1)$.
    \item Return to step 1.
\end{enumerate}
The correctness of this procedure follows from the DSU maintaining the invariant that the representative for a range is the first range that succeeds it whose successor has not been set yet.
\end{proof}

\begin{figure}[h]
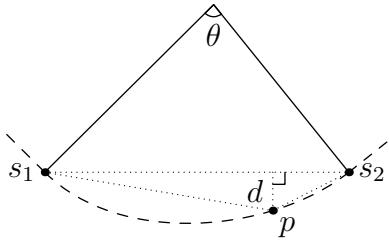

    \centering
    \begin{asy}
    import olympiad;
    
    unitsize(1cm);
    pair s1 = (-2,0);
    pair s2 = (2,0);
    dot("$s_1$",s1,W);
    dot("$s_2$",s2);
    pair x1 = (-1,-1);
    pair x2 = (1,-0.8);
    draw(s1..controls x1 and x2..s2, dashed);
    draw(s1--(s1+0.5*(s1-x1)),dashed);
    draw(s2--(s2+0.5*(s2-x2)),dashed);
    pair p = intersectionpoints((1,-1)--(1,1),s1..controls x1 and x2..s2)[0];
    dot("$p$",p,SE);
    draw(s1--s2, dotted);
    pair f = foot(p,s1,s2);
    draw(f--p,dotted);
    label("$d$",(p+(1,0))/2,W);
    pair orig = extension(s1,s1+(x1-s1)*dir(90),s2,s2+(x2-s2)*dir(-90));
    draw((-2,0)--orig--(2,0));
    draw(s1--p--s2,dotted);
    markscalefactor=0.02;
    draw(rightanglemark(p,f,s2));
    draw(anglemark((-2,0),orig,(2,0)));
    label("$\theta$",orig,2*S);
    \end{asy}
    \caption{Diagram for \Cref{lemma:polyapprox1}. The dashed line is part of $\partial C$ while the solid lines denote the normals to $\partial C$ at $s_1$ and $s_2$ respectively.}
    \label{fig:polyapprox1}
\end{figure}

\begin{proof}[Proof of \cref{lemma:half-eps-intersect}]First, see \cref{lemma:polyapprox1} for the case where $C$ lies strictly within $U$. Next, we describe what needs to be changed when this is not the case:

If we have exact access to $\partial(U)\cap C$, then we can convert a closest point oracle for $C$ into a closest point oracle for $C\cap U$ and apply the previous lemma to select points on $\partial (C\cap U)$. Note that due to the convexity of $C$, the intersection of each side of $U$ with $C$ is a line segment if it is nonempty, and the endpoints of this line segment can be approximated to arbitrary precision using the oracle to binary search. 

Specifically, if $C$ intersects a side $s\subseteq \partial(U)$, we can approximate this intersection by binary searching on $s$. Say we take some $p\in s$ such that $p\not\in C$; then the direction of the vector from $p$ to $\closest_C(p)$ tells us which side of $s\cap C$ $p$ lies on). After $\BIGO{\log \frac{1}{\eps}}$ queries, we either end up with:
\begin{enumerate}
    \item A point within $C\cap s$, and approximations of both endpoints of $C\cap s$ to within distance $o(\epsilon)$ each.
    \item A segment of length $o(\epsilon)$ containing $C\cap s$.
\end{enumerate}
In either case, given an estimate of an endpoint of $C\cap s$ that is within $o(\epsilon)$ of an endpoint of $C\cap s$, by projecting the estimate onto $C$ we get an estimate of the endpoint that is within $o(\epsilon)$ to the true endpoint (and now is part of $C$).

When querying the closest point in $C\cap U$ to a point $p$, 
\begin{enumerate}
    \item Compute the closest point in $U$ to $p$. If the point is within $C$, return it.
    \item Otherwise, query the oracle for the closest point in $C$ to $p$. If the point is within $U$, return it.
    \item Otherwise, return the closest estimated endpoint (which is guaranteed to be within $o(\epsilon)$ of the true answer). \qedhere
    
\end{enumerate}


\end{proof}

\begin{lemma}[Polygonal Approximation (Simpler)]

\label{lemma:polyapprox1}
Given a closest point oracle for a convex region $C$ that is strictly contained within a unit square $U$, we may select $\BIGO{\eps^{-1/2}}$ points on $\partial C$ such that every point within $C$ is within distance $\eps$ of the convex hull of the selected points.
\end{lemma}

\begin{proof}
Let $S$ denote the set of all selected points. First, we state a sufficient pair of conditions for $S$ to satisfy the desired property: For every pair of consecutive points $s_1$ and $s_2$ in $S$, 
\begin{enumerate}
    \item The distance between $s_1$ and $s_2$ along the border of $C$ is at most $\sqrt \eps$.
    
    \item The difference in angle $\theta$ between the normals to $C$ at $s_1$ and $s_2$ is at most $\sqrt \eps$.
\end{enumerate}

To see that this is true, consider any point $p$ on $\partial C$ lying between $s_1$ and $s_2$, and define $d\triangleq \text{dist}(p,s_1s_2)$. See \Cref{fig:polyapprox1} for an illustration. Then 
\begin{align}
\frac{d}{\sin \angle ps_1s_2}+\frac{d}{\sin \angle ps_2s_1}=|s_1-s_2|&\implies \frac{d}{\sin \angle s_1ps_2}< |s_1-s_2| \nonumber \\
&\implies \frac{d}{\sin \theta} < |s_1-s_2| \label{ineq:sin} \\
&\implies \frac{d}{\sqrt \eps}<\sqrt \eps \implies d<\eps \nonumber
\end{align}

Inequality (\ref{ineq:sin}) follows from $m\angle s_1ps_2>\pi-\theta$.

It remains to construct $S$ satisfying both of the desired conditions. Note that we can actually construct two separate sets of points $S_1$ and $S_2$, one for each of the two conditions, and then set $S=S_1\cup S_2$. Let $S'$ denote a set of $4\ceil{\frac{1}{\sqrt \eps}}$ points spaced equally about the border of $U$.

\begin{enumerate}
    \item Form $S_1$ by projecting each point in $S'$ onto $C$.
    
    \item Form $S_2$ by dilating each point in $S'$ by a sufficiently large constant about the center of $U$ and then projecting each of these points onto $C$. Essentially, we use the closest point oracle to implement a tangent line oracle.  There is an alternate proof of \cref{lemma:polyapprox1} that only involves querying the closest point oracle at points on the border of $U$.
\end{enumerate}

Since the distance between any two consecutive points in $S'$ is at most $\frac{1}{\sqrt \eps}$, $S_1$ satisfies condition 1 by a stronger version of \Cref{lemma:project-contract}. Furthermore, condition 1 continues to be satisfied as points from $S_2$ are added to $S$. Similarly, it's easy to verify that $S_2$ satisfies condition 2, and that condition 2 continues to be satisfied as points from $S_1$ are added to $S$.
\end{proof}

\begin{proof}[Proof of \cref{lemma:polyapprox1} (Alternative)]

There is no need to query the closest point oracle for points that can be infinitely far away from $C$ (which is done in the original proof). Let's start with a set $S'$ containing a single point on the border of $U$. While there are two consecutive points $s_1'$ and $s_2'$ in $S'$ such that the clockwise distance between $s_1'$ and $s_2'$ along the border of $U$ times the difference in angle between the normals to $C$ at $\closest(s_1')$ and $\closest(s_2')$ is greater than $\eps$, insert an additional point $m$ into $S$ such that $m$ is the midpoint of the portion of $\partial(U)$ that goes clockwise about $U$ from $s_1'$ and $s_2'$. Once no such pair of points exists in $S'$, set $S=\{\closest(s')\mid s'\in S'\}$.

Here is a different way to interpret this process. Start with a list of pairs initially containing only $(4,2\pi)$, corresponding to the length of the border of $U$ and the measure of a full angle, respectively. While there exists a pair $(x,y)$ satisfying $xy>\eps$, remove it from the list and add the pairs $(x/2,r)$ and $(x/2,y-r)$ to the list, where $r\in [0,y]$. Here,

\begin{itemize}
    \item $x$ represents an upper bound on the distance from $\closest(s_1')$ to $\closest(s_2')$ along $\partial(C)$, and $x/2$ is an upper bound on both the distances from $\closest(s_1')$ to $\closest(m)$ and the distances from $\closest(m)$ to $\closest(s_2')$ along $\partial(C)$.
    \item $y$ is the difference in angle between the normals to $C$ at $\closest(s_1')$ and $\closest(s_2')$. Adding $m$ in between $s_1'$ and $s_2'$ splits this angle into two parts.
\end{itemize}

The correctness of the stopping condition can be proved similarly to the original proof of \Cref{lemma:polyapprox1}.  It remains to prove that the size of the list upon the termination of this process is $\BIGO{\frac{1}{\sqrt \eps}}$. Define $\mathit{potential}(x,y)=\max(1,4\sqrt{\frac{xy}{\eps}}).$ We claim that $\mathit{potential}(4,2\pi)\le \BIGO{\frac{1}{\sqrt \eps}}$ is an upper bound on the size of the list upon termination. It suffices to show that whenever $xy>\eps$, the following inequality holds for any choice of $r$:
\begin{equation*}
\mathit{potential}(x,y)\ge \mathit{potential}(x/2,r)+\mathit{potential}(x/2,y-r).
\end{equation*}

This may be rewritten as:
\begin{equation*}
 4\sqrt{\frac{xy}{\eps}}\ge \max\paren{2,4\sqrt{\frac{x/2\cdot y}{\eps}}+1,4\sqrt{\frac{x/2\cdot r}{\eps}}+4\sqrt{\frac{x/2\cdot (y-r)}{\eps}}}  
\end{equation*}

which can easily be verified. Equality holds when $r=y/2$.
\end{proof}

\subsection{Disjoint convex fat bodies: omitted proofs}\label{sec:omitted-proofs-convex-fat}

\begin{proof}[Proof of \cref{lemma:packing-lb}]
	We first see that it suffices to find an $n$ such that the optimal path must be $\Omega(\min r_h)$.
	Once we find such an $n$, we can use the bound on subsequences of size $n$ to obtain the desired result.

	Let $m\triangleq\min r_h$. Suppose that $OPT<m$; it suffices to show that $n= \BIGO{1}$.

	Let $C(P, r)$ be the sphere centered at $P$ with radius $r$.
	Pick an arbitrary point $p$ on $OPT$. Observe that all of $OPT$ lies inside $C(p, m)$, and for each region $R$, there is a point $q\in R\cap C(p,m)$. 
	
	\begin{claim*}
		
		For all convex fat regions $R$ with $r_h\geq m$, we must have
		$\vol(C(q,m)\cap R)\geq \Omega\left(\vol(C(q,m))\right)$.
		
	\end{claim*}
	
	If the claim holds, since the regions $C(q,m)\cap R$ are disjoint subsets of $C(p,2m)$,
	we would have $n\cdot \Omega\left(\vol(C(q,m))\right)\leq \vol(C(p,2m))$,
	implying the desired result $n= \BIGO{1}$.

	\begin{claimproof}
		Let $R_c$ be the center of a ball contained in $R$ with radius $r_h$.
		Consider the case that $\vmag{R_c-q}> \frac 12 m$.
		By convexity, since $q\in R$ and $\vol(C(R_c,r_h))\subseteq R$,
		the image of a dilation of $\vol(C(R_c,r_h))$ with center $q$ and ratio $\frac {m}{2\vmag{R_c-q}}\leq 1$
		is also a subset of $R$.
		Let image of the dilation be $\vol(C(R_c',r_h'))$.
		We have that $\vmag{q-R_c'}=\frac 12 m$ and
		$r_h'=r_h\cdot \frac {m}{\vmag{R_c-q}}\geq r_h\cdot \frac {m}{r_H}\geq \Omega(m)$.
		Therefore, it suffices to prove that
		$\vol(C(q,m)\cap C(R_c, \Omega(m)))\geq \Omega\left(\vol(C(q,m))\right)$
		for all points $\vmag{q-R_c}\leq \frac 12 m$.

		WLOG $\Omega(m)\leq \frac 12 m$. Then $C(R_c, \Omega(m))$ lies entirely inside $\vol(C(q,m))$, so
		$\vol(C(q,m)\cap C(R_c, \Omega(m)))= \vol( C(R_c, \Omega(m))) \geq \Omega\left(\vol(C(q,m))\right)$,
		as desired.
	\end{claimproof}

	\begin{claim*}
		The bound for balls is $n_0=3$.
	\end{claim*}

	\begin{claimproof}
		First, we reduce this claim to the 2D case.  More specifically, we want to show that in the case of three balls,
		the optimal path must lie on the plane $P$ containing the centers of the balls. 
		
		For any path $p$, the projection $p'$ of $p$ onto $P$ is also a valid path
		and has length at most the length of $p$. We note that projection never increases
		the length of a segment; therefore, the distances of the points on the path to
		the centers of their respective balls must have decreased. Therefore, $p'$ still passes through all three balls, and the length of $p'$ is at most that of $p$.

		It remains to show the claim in the 2D case. We claim that we must have $OPT\geq \frac1{100}\left( \min r_h\right)$. We proceed similarly to the general case:
		Assume the contradiction, and let $m\triangleq\frac1{10}\left( \min r_h\right)$.
		Then the optimal path lies in a disk $C\left(p, \frac{1}{10}m\right)$.
		For each region $R$, let $q$ be any point in $p\cap R$. Then $C\paren{q, m}$ is contained inside $C\paren{p, \paren{1+\frac{1}{10}}m}$.
		Note that this differs from the general case as we assumed $OPT<\frac{1}{10}m$ instead of
		$OPT<m$.

		Therefore, to create a contradiction, it remains to show that
		\begin{equation*}
		    \vol(C(q, m)\cap R)> \frac13 \vol\paren{C\paren{p, \paren{1+\frac{1}{10}}m}}=
		\frac{1}{3}\paren{1+\frac{1}{10}}^2\vol(C(p, m)),
		\end{equation*}
		where $R$ is any disk with radius at least $10m$ and $q\in R$.

		Let $r$ be the center of $R$, and $a$ be any point inside $C(q,m)$.
		Also let $x\triangleq \vmag{r-q}$, $y\triangleq \vmag{q-a}$, and $z\triangleq \vmag{r-a}$.
		By the Law of Cosines,
		$z^2=x^2 + y^2 -2xy \cos\paren{\angle{rqa}}$.

		Suppose that we have $\cos\paren{\angle{rqa}}\geq \frac{1}{20}$;
		we claim that this implies $z\leq 10m$.
		Note that from our constraints we have $x\leq 10m$ and $y\leq m\leq 10m$,
		so we have
		\begin{equation*}
		 (x-10m)(y-10m)=xy-10mx-10my+100m^2\geq 0.
		\end{equation*}
		Next, we have
		\[
			\begin{split}
				z^2
				&\leq x^2 + y^2
				  -\frac{1}{10}xy \\
				&\leq x^2 + y^2
				  -\frac{1}{10}\paren{10mx+10my-100m^2} \\
				&= \paren{x-\frac{1}{2}m}^2 + \paren{y-\frac12 m}^2
				  +\frac{19}{2}m^2 \\
				&\le \paren{10m-\frac{1}{2}m}^2 + \paren{m-\frac12 m}^2
				  +\frac{19}{2}m^2 \\
				&=100m^2.
			\end{split}
		\]
		To conclude,
		\begin{equation*}\vol(C(q, m)\cap R)\geq \frac{2\arccos{\paren{\frac1{20}}}}{2\pi}\vol(C(q, m))>
		\frac13 \vol\paren{C\paren{q, \paren{1+\frac{1}{10}}m}},
		\end{equation*}
		as desired.
	\end{claimproof}
\end{proof}

\begin{proof}[Proof of \cref{lemma:sp-logn} (Remainder)]

For the construction, we'll address the case where all the bodies are disks on the 2D plane; the result can trivially be extended to higher dimensions. Let $x_i$ be the $x$-coordinate of the center of the $i$th largest disk. We will show that it is possible to have $r_i=\frac1i$ for all $i$ such that every disk is tangent to the segment connecting $(0,0)$ and $(8,0)$, every disk has center above the $x$-axis, and no two disks intersect. As $\sum r_i=\Theta(\log n)$ and $OPT\le 8$, this would give the desired bound.

We claim that regardless of how $x_1,x_2,\ldots,x_{i-1}$ have been selected, there is always a valid choice for $x_i$ such that the $i$th disk does not intersect with any of the first $i-1$. Observe that $x_i$ is valid if $|x_i-x_j|\ge \sqrt{(r_i+r_j)^2-(r_i-r_j)^2}= 2\sqrt{r_ir_j}$ for all $j\in [1,i-1]$, where $r_j$ denotes the radius of the $j$th disk. The total length of the $x$-axis rendered invalid by any of the first $i-1$ disks is at most
\begin{equation*}
\sum_{j=1}^{i-1}4\sqrt{r_ir_j}=\frac{4}{\sqrt i} \sum_{j=1}^{i-1}\frac{1}{\sqrt j}< \frac{8\sqrt i}{\sqrt i}< 8.
\end{equation*}

Therefore, some $x\in [0,8]$ must exist that was not rendered invalid and is thus a valid candidate for $x_i$.
\end{proof}

\subsection{Balls: omitted proofs} \label{sec:omitted-proofs-balls}

\begin{proof}[Proof of \cref{lemma:local-global-opt}]
As the problem of touring disks be formulated as a convex optimization problem by \cref{lemma:socp} and the optimal value is lower bounded, a global optimum is guaranteed. This global optimum must be locally optimal, in the sense that it should not be able to decrease the objective by moving any single $p_i$. This means that
\begin{enumerate}
    \item For all $p_i$ satisfying $|p_i-c_i|<r_i$, the gradient of the objective with respect to $p_i$ must be 0.
    \item For all $p_i$ satisfying $|p_i-c_i|=r_i$, the gradient of the objective with respect to $p_i$ must be perpendicular to the $i$th circle.
\end{enumerate}
The gradient of the objective with respect to $p_i$ is precisely $\unit(p_i-p_{i-1})+\unit(p_i-p_{i+1})$, where $\unit(x)\triangleq \frac{x}{\vmag{x}}$. Case 1 corresponds to the tour passing straight through the $i$th disk while case 2 corresponds to reflecting off the $i$th circle.

Conversely, given a locally optimal solution, we can certify its optimality by choosing $z_i=\unit(p_{i+1}-p_i)$, where $z_i$ is defined in the proof of \cref{lemma:socp}.
\end{proof}

\begin{lemma}\label{lemma:socp}
The touring balls problem can be formulated as a convex optimization problem (specifically, a second-order cone problem).
\end{lemma}

\begin{proof}
Stated in \cite{polishchuk2005touring}. We restate a possible formulation here:

\proofsubparagraph*{Primal.}

\noindent
Constants: $c_i \in \R^d, \forall i\in [0,n+1]$. $r_i \in \R^+, \forall i\in [0,n+1]$.

\noindent
Decision Variables: $p_i \in \R^d, \forall i\in [0,n+1]$. $d_i \in \R^+, \forall i\in [0,n]$.

\noindent
Constraints: $\vmag{p_i-c_i} \leq r_i, \forall i \in [0, n+1]$. $\vmag{p_{i+1}-p_{i}} \leq d_i, \forall i \in [0, n]$.

\noindent
Objective: $\min\paren{\sum_{i=0}^nd_i}$

\proofsubparagraph*{Dual.}

\noindent
Constants: Same as primal.

\noindent
Decision Variables: 
\begin{itemize}
    \item Associate a variable $y_i\in \R^d, \forall i\in [0,n+1]$ and a real $w_i\in \R^+, \forall i\in [0,n+1]$ for each constraint of the first type ($y_i\cdot (p_i-c_i)\le w_i\cdot r_i$).
    \item Associate a variable $z_i\in \R^d, \forall i\in [0,n]$ for each constraint of the second type ($z_i\cdot (p_{i+1}-p_i)\le d_i$).
\end{itemize}

\noindent
Constraints: $\vmag{z_i}\le 1, \forall i\in [0,n]$. $y_i=z_i-z_{i-1}, \forall i\in [0,n+1]$. $\vmag{y_i}\le w_i, \forall i\in [0,n+1]$. 

\noindent
Objective: $\max\paren{-\sum_{i=0}^{n+1}w_ir_i-\sum_{i=0}^{n+1}y_i\cdot c_i}$.
\end{proof}
\end{document}